\documentclass[a4paper,USenglish,thm-restate]{lipics-v2021}
\nolinenumbers
\hideLIPIcs
\usepackage[utf8]{inputenc}
\usepackage{algorithm}
\usepackage[noend]{algpseudocode}
\usepackage{mdframed}
\usepackage{amssymb}
\usepackage{amsmath}
\usepackage{amsthm}
\usepackage{graphicx}
\usepackage{cite}
\usepackage{mathtools}
\usepackage{comment}  
\usepackage{makecell}

\usepackage{color,soul}
\newcommand{\similar}[1]{\stackrel{#1}{\sim}}
\newcommand\change[1]{{\color{red}#1}}  

\title{Why Canonical Rounds Fail for Optimal Byzantine Resilience}
\author{Hagit Attiya}{Technion, Haifa 3200003, Israel}{}{}{}

\author{Itay Flam}{Technion, Haifa 3200003, Israel}{}{}{}

\author{Jennifer L.\ Welch}{Texas A\&M University, College Station, TX 77843-3112, USA}{}{}{}

\authorrunning{Attiya, Flam and Welch}
\titlerunning{Why Canonical Rounds Fail for Optimal Byzantine Resilience}

\ccsdesc[500]{Theory of computation~Distributed algorithms}

\keywords{
randomized consensus,
approximate agreement,
crusader agreement,
connected consensus,
reliable broadcast,
gather}

\begin{document}

\maketitle 

\begin{abstract}
Canonical asynchronous rounds are a widely used abstraction for structuring distributed algorithms, making asynchronous executions appear synchronous and enabling modular reasoning. 
We show that this abstraction is fundamentally incompatible with optimal resilience in the Byzantine setting, even when randomization is allowed. 
Specifically, we prove that when $3f < n \le 5f$, where $n$ is the number of processes and at most $f$ may be Byzantine faulty, no randomized canonical-round algorithm can solve consensus with bounded expected round complexity, and that communication-closed variants fail to solve consensus altogether in this regime.

We establish these lower bounds via a unifying notion of \emph{nontrivial convergence}, which captures consensus as well as classical relaxations, such as approximate agreement and connected consensus. 
Using simple reductions, the same impossibility extends to fundamental communication primitives such as reliable broadcast and gather.
Our results identify a sharp boundary: while bounded canonical-round algorithms for these problems exist when $n > 5f$, they cannot when $n \le 5f$.  
Thus optimal resilience, $n > 3f$, cannot be achieved within the canonical-round framework.

On the positive side, we show that the gather primitive captures the content-dependent communication needed to bypass this limitation. 
We use gather to obtain a simple and modular algorithm for connected consensus with optimal resilience, clarifying the communication structures required for optimal resilience. 
\end{abstract}

\ccsdesc[500]{Theory of computation~Distributed algorithms}

\keywords{randomized consensus, approximate agreement, crusader agreement, connected consensus, reliable broadcast, gather}

\section{Introduction} 
\label{sec:introduction}

Consensus is the fundamental coordination problem in distributed
computing.  In asynchronous systems with Byzantine failures,
deterministic consensus is impossible, but randomized consensus is
solvable with optimal resilience $n > 3f$, 
where $n$ is the number of processes and 
$f$ is the maximum number of faulty processes.  
Randomized Byzantine
consensus under optimal resilience lies at the core of many
fault-tolerant systems, including modern blockchain protocols.

Ben-Or's pioneering randomized consensus algorithm~\cite{BenOr1983}
uses a simple communication pattern called \emph{canonical
(asynchronous) rounds}~\cite{Fekete1994asynchronous},
in which processes repeatedly exchange round-tagged
messages, wait for $n-f$ messages from the current round, and then
advance.  Besides its conceptual simplicity, counting the number of
asynchronous rounds is a natural way to analyze performance.
However, Ben-Or's algorithm requires $n > 5f$ and takes
an exponential expected number of rounds, both of which are sub-optimal.

Work on developing improved algorithms for randomized consensus continues to
this day, with much emphasis on achieving optimal resilience of $n >
3f$.  These algorithms are sophisticated, and their communication
patterns are more involved than that in~\cite{BenOr1983}.  
For instance, Canetti and Rabin's algorithm~\cite{CanettiR1993} runs in
constant expected time
and requires $n > 3f$, but is an intricate combination of primitives.
One of these primitives is
Bracha's \emph{reliable broadcast} algorithm~\cite{Bracha1987}, 
which prevents faulty processes from equivocating by ensuring that correct processes agree on the value sent by a given sender.  
This algorithm,
reproduced as Algorithm~\ref{alg:bracha}, relies on content-dependent
wait conditions that are not satisfied by merely advancing through
canonical rounds.  For example, Step~1 is not complete until
various numbers of various kinds of messages have been received.

\begin{algorithm}[tb]
  \caption{Reliable Broadcast ($n>3f$), from~\cite{Bracha1987}}
  \label{alg:bracha}
  \small
  \begin{algorithmic}[1]
    \State sender sends $(\mathit{initial},v)$ to all processes
           \Comment{\textbf{Step 0}, triggered by r-broadcast($v$)}
    \State wait until the receipt of:
        \Comment{\textbf{Step 1}} \label{line:bra:step 1}
    \Statex \quad 
        one $(\mathit{initial}, v)$ message
        or $(n+f)/2$ $(\mathit{echo},v)$ messages
        or $f+1$ $(\mathit{ready},v)$ messages 
        for some $v$  
    \State send $(\mathit{echo},v)$ to all the processes.
    \State wait until the receipt of:
        \Comment{\textbf{Step 2}} \label{line:bra:step 2}
    \Statex \quad 
        $(n+f)/2$ $(\mathit{echo},v)$ messages
        or $f+1$ $(\mathit{ready},v)$ messages 
        (including ones received in step 1) 
    \Statex \quad\quad\quad for some $v$ 
    \State send $(\mathit{ready},v)$ to all the processes
    \State wait until the receipt of: 
        \Comment{\textbf{Step 3}} \label{line:bra:step 3}    
    \Statex \quad 
        $2f+1$ $(\mathit{ready},v)$ messages 
        (including ones received in steps 1 and 2) 
        for some $v$
    \State \textbf{r-accept} $v$
  \end{algorithmic}
\end{algorithm}

This raises a natural question: \emph{can optimal-resilience randomized consensus be achieved within the simple structure of canonical rounds?}
A strawman attempt to translate Algorithm~\ref{alg:bracha} into canonical rounds is shown in Algorithm~\ref{alg:bracha canon}.
Unfortunately, it can experience an unbounded number of rounds:
Even with a correct sender, the
sender's {\it initial} message can be delayed arbitrarily long, while
canonical rounds consisting largely of {\it null} messages continue to
advance.  
In contrast, the \emph{time complexity} of this algorithm is constant,
when measured as the real time elapsed
until termination assuming every message sent between correct
processes is delivered within one time unit~\cite{AttiyaW2004}:
Though Byzantine processes may ``rush'' correct processes through many
canonical rounds, as long as the sender's {\it initial} message is
in transit, less than one time unit can elapse.
Thus canonical round complexity can diverge arbitrarily from real time: the algorithm terminates in constant time while requiring an unbounded number of canonical rounds.

\begin{algorithm}[bt]
  \caption{Canonical-Round Reliable Broadcast ($n>3f$),
    based on~\cite{Bracha1987}; code for process $p_i$}
  \label{alg:bracha canon}
  \small
  \begin{algorithmic}[1]
  \Statex local variable {\it round}, initially 1
  \State {\bf WakeUp:} \Comment{triggered by r-broadcast$(v)$ at the sender and r-broadcast$(-)$ at other processes}
  \If{$p_i$ is the sender} 
      send $(${\it initial}$,v)$ labeled with round 1 to all processes
  \Else ~send$(${\it null}$)$ labeled with round 1 to all processes \textbf{endif}
  \EndIf
  \State {\bf Receive message $m$ with round label $r$:}
  \If{$r =$ {\it round} and $p_i$ has now received $n-f$ round $r$ messages} 
  \State $S := \emptyset$  \Comment{set of messages to send for round $r+1$}
    \If{$p_i$ has received 
         $(${\it initial}$,v)$ or 
        $(n+f)/2$ $(${\it echo}$,v)$'s or 
        $f+1$ $(${\it ready}$,v)$'s 
    \Statex \quad\quad\quad and not sent $(${\it echo}$,v)$} 
        add $(${\it echo}$,v)$ to $S$ \textbf{endif} \EndIf
  \If{$p_i$ has received 
        $(n+f)/2$ $(${\it echo}$,v)$'s or 
        $f+1$ $(${\it ready}$,v)$'s 
  \Statex \quad\quad\quad and not sent $(${\it ready}$,v)$} 
        add $(${\it ready}$,v)$ to $S$ \textbf{endif} \EndIf
  \If{$p_i$ has received $2f+1$ $(${\it ready}$,v)$'s 
  \Statex \quad\quad\quad and not done r-accept}
        {\bf r-accept} $v$ \textbf{endif} \EndIf
  \State {\it round} $:=$ {\it round} $+ 1$
  \If{$S = \emptyset$} $S := {(null)}$ \textbf{endif} \EndIf
  \State send messages in $S$ labeled with current value of {\it round} variable to all processes
  \Statex \textbf{endif} \EndIf
 \end{algorithmic}
\end{algorithm}

This bad behavior is not simply an artifact of poor translation.  We
identify a fundamental limitation of canonical round-based
abstractions in the Byzantine setting.  \emph{When $3f < n \le 5f$, no
randomized canonical-round algorithm can solve Byzantine consensus
with bounded expected round complexity.  
Moreover, if communication closure is imposed, which requires early and late messages to be discarded, canonical-round algorithms fail to solve consensus altogether in this resilience regime.}

These results are not restricted to randomized consensus.  Our theorems
are proved for a generalization of consensus that we call
\emph{nontrivial convergence}.  
This class of problems includes \emph{approximate agreement} 
(both on the real numbers and on graphs) and
\emph{connected consensus} (including its special cases 
\emph{crusader agreement} and \emph{gradecast}).  
Unlike consensus, these problems have deterministic solutions in asynchronous systems.
The deterministic versions of our lower bounds imply that no optimally-resilient canonical-round algorithm for any of these problems can have bounded worst-case round complexity, and no communication-closed algorithms exist.

We extend these limitations to two communication primitives,
reliable broadcast and gather, via reductions from crusader agreement.
The reduction to reliable broadcast justifies the poor behavior of
Algorithm~\ref{alg:bracha canon}.  
Canonical rounds enforce progress based solely on \emph{message counts}, whereas optimally-resilient Byzantine algorithms require progress conditions that depend on \emph{message contents}.
This is exemplified with \emph{gather}, which captures a form of
content-dependent communication that canonical rounds exclude:
informally, gather allows each process to collect a large set of
values such that all correct processes share a substantial
\emph{common core}.  
There are existing gather implementations with optimal resilience and constant time complexity~\cite{SternA2021}
(see~Appendix~\ref{app:gather}).

We extend the reduction used to establish the impossibility results
for gather into a simple and modular algorithm for connected
consensus for any parameter~$R$, yielding an optimally-resilient
algorithm whose running time is logarithmic in~$R$.  An additional
advantage of this reduction is that it preserves the \emph{binding}
property~\cite{AbrahamBDY2022}: once the first correct process
decides, the adversary cannot influence the decisions of other correct
processes.  When gather satisfies binding, the resulting
connected-consensus algorithm inherits this property directly. 

In summary, this paper makes the following contributions relating to
asynchronous systems subject to Byzantine failures.
\begin{itemize}
\item 
We show that no randomized canonical-round algorithm can solve a nontrivial convergence problem with bounded expected round complexity when $3f < n \le 5f$, and that communication-closed canonical-round algorithms fail completely in this regime.  
\emph{Thus there can be no
such algorithms for randomized consensus that have optimal resilience
$n > 3f$.}

\item The deterministic versions of our lower bounds imply that no canonical-round algorithms have bounded worst-case round complexity, and no communication-closed canonical-round algorithms exist at all, for approximate agreement and $R$-connected consensus that have optimal resilience $n > 3f$.

\item Simple reductions from crusader agreement extend our lower bounds to two popular communication primitives, gather and reliable broadcast.

\item The reduction from crusader agreement to gather forms the basis of
a modular $R$-connected consensus algorithm, for any $R \ge 1$,
that has optimal resilience and preserves the binding property.
Beyond highlighting the usefulness of gather,
the resulting algorithm is of independent interest.
\end{itemize}

\subsection*{Related Work} 
Round-based abstractions, such as
canonical rounds~\cite{Fekete1994asynchronous}, 
the \emph{Heard-Of model}~\cite{CharronBostSchiper2009}
and communication-closed layers~\cite{ElradF1982}, 
make asynchronous executions appear
synchronous which enables modular reasoning and intuitive running time
calculations.  
In the canonical round structure, processes must move to the next
round as soon as a fixed \emph{number} of messages have been received,
regardless of the message contents or the identities of the senders.
Several well-known Byzantine-tolerant algorithms are
structured in canonical rounds, including ones for randomized
consensus~\cite{BenOr1983}, approximate
agreement~\cite{DolevLPSW1986}, and crusader
agreement~\cite{AttiyaW2023}, albeit assuming $n > 5f$.
In contrast, many optimally resilient Byzantine algorithms rely
critically on content-dependent communication patterns, including
those for randomized consensus~\cite{CanettiR1993}, crusader
agreement~\cite{Dolev1982,AbrahamBDY2022,AbrahamBDSY2023} and its
generalization to connected consensus~\cite{AttiyaW2023}, approximate
agreement~\cite{Coan1988,AbrahamAD2004},
gather~\cite{CanettiR1993,SternA2021}, and reliable
broadcast~\cite{Bracha1987}.  
Our results explain the dichotomy between algorithms that require 
$n>5f$ and those that achieve optimal resilience $n>3f$.

Round-based abstractions have played an important role in
structuring, analyzing, and verifying distributed algorithms, and have
been extended to Byzantine transmission faults~\cite{BielyEtAl2007}.
These abstractions are particularly effective for benign fault models~\cite{AttiyaCGT2025}
and for systematic testing and
verification~\cite{damian2019communication,DragoiEtAl2020}.  Our
impossibility results clarify their limitations in the Byzantine
setting: while useful for reasoning and verification, they cannot
capture the communication patterns required for optimal-resilience
Byzantine coordination.

In contrast to some early work which gave the impression implicitly
that asynchronous algorithms could be structured as canonical rounds
without any loss of generality
(e.g.,~\cite{DolevLPSW1986,ChorD1989,Fekete1994asynchronous}), later
authors indicated some skepticism on this point
(e.g.,~\cite{CharronBostSchiper2009}).  
Lewko~\cite{Lewko2011arxiv} in
particular notes that it ``is also worth considering if the seemingly
natural notion of a round (imported from the synchronous setting) may
have a restrictive effect on our thinking in the asynchronous
setting.'' Our results show that the notion of round is indeed
restrictive when optimal Byzantine resilience is desired.

\section{Preliminaries}
\label{sec:model}
In this section, we present our model of computation.
We also define a generic problem, called ``nontrivial convergence'', and show that several well-known problems are special cases of it.
We then define canonical-round algorithms and present our complexity measures.

\subsection{Model of Computation}

We assume the standard asynchronous model for $n$ processes, up to $f$ of which can be faulty, in which processes communicate via reliable point-to-point messages.
We consider \emph{Byzantine} failures, where a faulty process can change state arbitrarily and send messages with arbitrary content.

In more detail, we assume a set of $n$ processes, each modeled as a
state machine.
Each process has a subset of initial states, with one state corresponding to
each element of a set $V$, denoting its input.
The transitions of the state machine are triggered by 
two kinds of {\em events}:  spontaneous WakeUp and receipt of a message.
A transition takes the current state of the process, a random number, and an incoming message (if any),
and produces a new state of the process and a set of messages to be
sent to any subset of the processes.
The state set of a process contains
a collection of disjoint subsets,
each one modeling the fact that a particular decision has been taken;
once a process enters the subset of states for a specific decision,
the transition function ensures that it never leaves that subset.

A \emph{configuration} of the system is a vector of process states,
one for each process, and a set of in-transit messages.
In an \emph{initial configuration}, each process is in an initial state and no
messages are in transit.
Given a subset of at most $f$ processes that are \emph{faulty} with the rest being \emph{correct}, we define an {\em execution} as a sequence $\alpha$ of alternating
configurations and events $C_0, e_1, C_1,\ldots$ such that:
\begin{itemize}
\item $C_0$ is an initial configuration.
\item The first event for each process is WakeUp.  
      A correct process experiences exactly one WakeUp 
      and a faulty process can experience any number of WakeUps. 
      The WakeUp can either be spontaneous (e.g., triggered by the invocation of the algorithm)
      or in response to the receipt of a message.
\item Suppose $e_i$ is an event in which process $p$ receives message $m$ sent by process $q$.
      Then $m$ is an element of the set of in-transit messages in $C_{i-1}$ and it is the oldest
      in-transit message sent by $q$ to $p$. This means that point-to-point links are FIFO.
\item Suppose $e_i$ is a step by correct process $p$
  and let $s$ and $M$ be the state and set of messages resulting from $p$'s
  transition function applied to $p$'s state in $C_i$, a random number, and, 
  if $e_i$ is a receive event, the message $m$ being received.
  Then the only differences between $C_i$ and $C_{i+1}$ are that in $C_{i+1}$, $m$
  is no longer in transit, $M$ is in transit, and $p$'s state is $s$.
  If $p$ is Byzantine, then $s$ and $M$ can be anything.
\item If $\alpha$ is infinite, then every correct process takes infinitely many steps and every message sent by a process to a correct process is eventually received.
\end{itemize}

We model the variability in asynchronous executions, other than that caused by random numbers,
by an \emph{adversary}, which is a function that takes the current prefix of the execution and returns the next event to occur.  
The adversary is subject to the constraints indicated above.
The adversary is \emph{full information} as it can see the random choices made so far by the processes which are reflected in the process states and message contents, and \emph{static} as it chooses the faulty processes initially.

We say that executions $\alpha$ and $\beta$ are {\it indistinguishable} to a set of processes $X$, denoted $\alpha \similar{X} \beta$, if, for each process $p$ in $X$, $p$ has the same initial state and experiences the same sequence of events and random numbers in $\alpha$ as in $\beta$.

To measure time complexity in an asynchronous message-passing system, we adopt the definition in~\cite{AttiyaW2004}:
We start by defining a timed
execution as an execution in which nondecreasing nonnegative integers
(``times'') are assigned to the events, with no two events by the same
process having the same time.  For each timed execution, we
consider the prefix ending when the last correct process decides, and
then scale the times so that the \emph{maximum time} that elapses between the
sending and receipt of any message between correct processes
is 1.  We define the time complexity as the maximum, over all such scaled timed
execution prefixes, of the time assigned to the last event
minus the latest time when any (correct) process wakes up.
We sometimes assume, for simplicity, that the first WakeUp event of each process
occurs at time 0.

\subsection{Nontrivial Convergence Problems}

Our lower bound results are
proved for a generic \emph{nontrivial convergence} problem 
in which there are at least two possible input values $x_0$ and $x_1$ and 
at least two decision values $d_0$ and $d_1$, such that:
\begin{itemize}
    \item {\bf Agreement:} if a correct process decides $d_0$ in an execution, 
    then no correct process can decide $d_1$ in the same execution. 
    \item {\bf Validity:} if all correct processes have input $x_i$, then every decision by a correct process is $d_i$, for $i = 0,1$.
\end{itemize}
For \emph{deterministic} algorithms, we also require:
\begin{itemize}
    \item {\bf Termination:} 
    All correct processes eventually decide.
\end{itemize}
When considering \emph{randomized} algorithms,
Agreement and Validity must always hold, 
but Termination is probabilistic.  
See Section~\ref{subsec:canonical-round-defs} for details.

{\bf Consensus}~\cite{PeaseSL1980} is clearly a nontrivial convergence problem.  In addition to termination, all 
decisions by correct processes must be the same (\emph{agreement}), and
if a correct process decides $v$, 
then some correct process has input $v$ (\emph{validity}).
Several well-known relaxations of consensus
also belong in this class, as we show next.

{\bf Approximate agreement on the real numbers} with
parameter $\epsilon > 0$~\cite{DolevLPSW1986} is defined as follows.
Processes start with arbitrary real numbers
and correct processes must decide on real numbers that are at most $\epsilon$ apart from each other (\emph{agreement}).
Decisions must also be contained in the interval of the inputs of correct processes (\emph{validity}). 
To show approximate agreement is a nontrivial convergence problem, choose any two real numbers whose difference is greater than $\epsilon$ as the two distinguished inputs and two distinguished decisions.  

{\bf Approximate agreement on graphs}~\cite{NowakR2019}
has each process start with a vertex of a graph $G$ as its input.
Correct processes must decide on vertices such that all decisions are within distance one of each other (\emph{agreement}) and inside the convex hull of the inputs (\emph{validity}).
When all processes start with the same vertex, validity implies they must decide on this vertex.
As long as the graph $G$ has two vertices that are at distance 2 apart, 
we can choose these vertices as the two distinguished inputs and two distinguished decisions to show that approximate agreement on $G$ is a nontrivial convergence problem.

{\bf Connected consensus}~\cite{AttiyaW2023} with parameter $R$ and input set $V$ ensures that correct processes decide on elements in $D = \{(v,g) | v \in V$ and $1 \le g \le R\} \cup \{\bot,0\}$; we refer to $v$ as the \emph{branch} and $g$ as the \emph{grade}.  
If two correct processes decide $(v,g)$ and $(v',g')$, then $|g-g'| \le 1$ and $v = v'$ unless one is $\bot$ (\emph{agreement}).
In addition, if $v \ne \bot$, then some correct process has input $v$, and if all correct processes have the same input, then $g = R$ (\emph{validity}).
To show that this is a nontrivial convergence problem, let $v$ and $v'$ in $V$ be the distinguished inputs and $(v,R)$ and $(v',R)$ be the distinguished decisions.
When $R = 1$,
we get {\bf crusader agreement}~\cite{Dolev1982}.
When $R = 2$,
we get {\bf gradecast}~\cite{FeldmanM1997}. 

\subsection{
Canonical-Round Algorithms and Their Round Complexity}
\label{subsec:canonical-round-defs}

In a {\emph canonical-round algorithm},
every message sent by a correct process consists of a round number
(a positive integer), the process' history (its initial state followed
by the sequence of random numbers and sequence of messages it has received so far), and a random
number drawn from a set $\mathcal{R}$.  
For deterministic algorithms, $|\mathcal{R}| = 1$ and thus we can 
drop that component of the messages.

Each correct process wakes up and sends to all processes its round 1 message, consisting of the label 1, its initial state, and (optionally)
a random number.
This step is called round 0 for the process.
All messages that arrive at a process are saved, including those
that arrive before it wakes up.
Once a process has received $n-f$ round $r$ messages, where $r$ is the
round number of the most recent message it sent, it sends its round $r+1$
message to all processes.
The interval starting immediately after a process sends its round $r$
message and ending when it sends its round $r+1$ message is called round $r$ for the process.
In each round, a process applies a function to its current history to see if it can decide yet. 
Even if it decides, it continues to send messages, in order to help other processes make progress toward decision.

Given an execution of a canonical-round algorithm, we define the (canonical) round complexity of the execution to be the smallest round $r$ such that every correct process decides in or before its round $r$.
The {\it round complexity} of a deterministic algorithm is the largest round complexity of any execution of the algorithm.

We next define the expected round complexity of a randomized algorithm (cf.~\cite{AttiyaW2004}).
An execution of a specific algorithm is uniquely determined by an adversary $\mathcal{A}$, an initial configuration $I$, and a collection of random numbers $\rho$ accessed by the processes, denoted $exec(\mathcal{A},I,\rho)$.
Given a predicate $P$ on executions, $\Pr[\mathcal{A},I,P]$ is the probability of $\{\rho : exec(\mathcal{A},I,\rho)$ satisfies $P\}$, i.e., the probability taken over the random numbers that the execution produced by $\mathcal{A}$ starting with $I$ satisfies $P$.
Let $T$ be a random variable on executions, in our case, the round complexity.  For a fixed $\mathcal{A}$ and $I$, the expected value of $T$, denoted $EX[\mathcal{A},I,T]$, is $\sum_x x \cdot \Pr[\mathcal{A},I,T=x]$.
The {\it expected round complexity} of the algorithm, denoted $EX[T]$, is the maximum, over all adversaries $\mathcal{A}$ and all initial configurations $I$, of $EX[\mathcal{A},I,T]$.

\section{Limitations of Canonical-Round Algorithms}
\label{sec:lb}

This section contains our lower bound results for
canonical-round algorithms concerning the number of rounds
(Section~\ref{subsec:lb-rounds}) and impossibility when communication-closed (Section~\ref{subsec:lb-cc}).

\subsection{Unbounded Number of Canonical Rounds}
\label{subsec:lb-rounds}

\begin{theorem}
\label{theorem:unbounded canonical rounds}
For any randomized canonical-round algorithm that solves the nontrivial convergence problem with $n \le 5f$,
the expected number of rounds until all correct processes decide is greater than $K$, for any integer $K \in \mathbb{N}$. 
\end{theorem}

\begin{proof}  Assume towards a contradiction that there exists a canon\-ical-round algorithm for nontrivial convergence with $n \le 5f$ whose expected
number of rounds until all correct processes decide is at most $K$.
For convenience, denote the specific values $x_0$, $x_1$, $d_0$, and $d_1$
in the definition of nontrivial convergence by $0$, $1$, $0$, and $1$
respectively.

For simplicity, we assume $n = 5f$, and divide the processes into five
disjoint sets of $f$ processes each:  $A$, $B$, $C$, $D$, and $E$.
We abuse notation and, for any two sets $X$ and $Y$, denote $X \cup Y$
by $XY$.

We consider the following initial configurations of the algorithm:
\begin{itemize}
\item Denote by $C_0$ the initial configuration in which processes in
$BCDE$ are correct and have input 0, while processes in $A$ are Byzantine.
\item Denote by $C_1$ the initial configuration in which processes in
$ABDE$ are correct and have input 1, while processes in $C$ are Byzantine.
\item Denote by $C_2$ the initial configuration in which processes in
$ABCD$ are correct, those in $BC$ have input 0 and those in $AD$ have input 1, 
while processes in $E$ are Byzantine.
\end{itemize}
For each $C_j$, let $C_j[i]$ denote $p_i$'s input in $C_j$.

We next define a set of executions of the algorithm that begin with $C_2$.
For each of these executions, we define a related execution that begins with $C_0$ and another one that begins with $C_1$.
All these executions are inductively constructed, round by round, with the behavior of the faulty processes in each one being dependent on the behavior of some correct processes in others.
(See Figures~\ref{figure:three executions} and~\ref{fig:combined}.)

\begin{figure}
\includegraphics[width=0.95\linewidth]{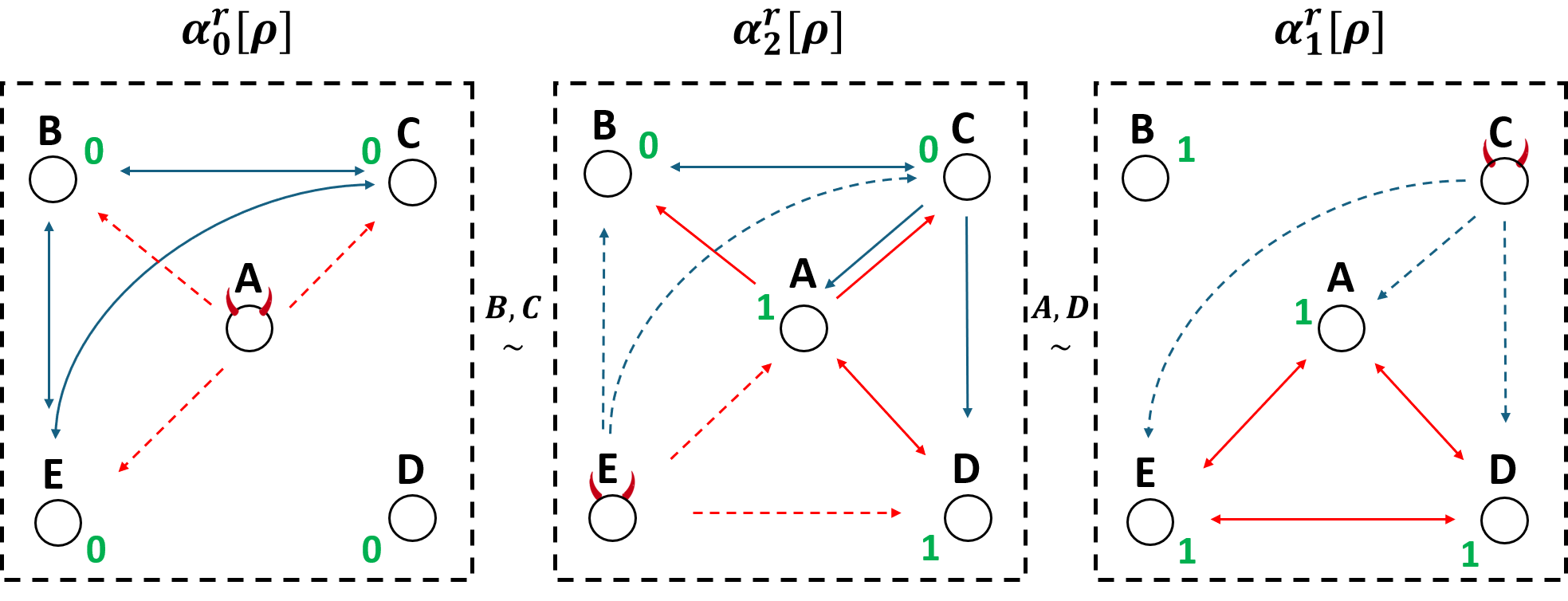}
\caption{Schematic description of the executions 
used in the proof of Theorem~\ref{theorem:unbounded canonical rounds}.
Messages are represented by directed arrows, with dotted arrows being those sent by Byzantine processes.
Blue arrows correspond to 0 inputs and red to 1.  
Note two-faced behavior of $E$ in $\alpha_2^r[\rho]$.
} 
\label{figure:three executions}
\end{figure}

\begin{figure}

    \centering
    \includegraphics[width=0.9\textwidth]{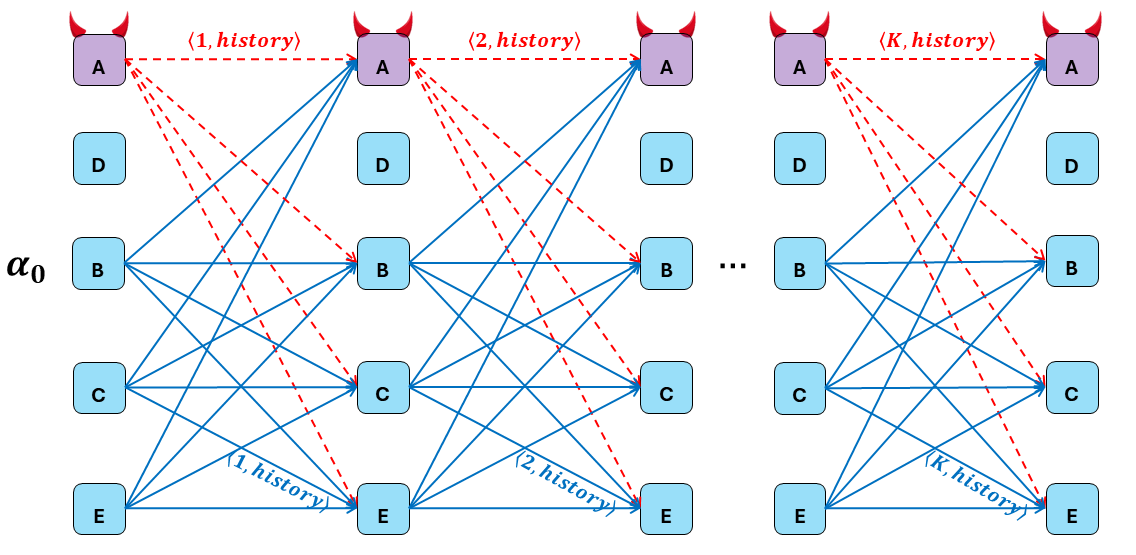}
    
    \includegraphics[width=0.9\textwidth]{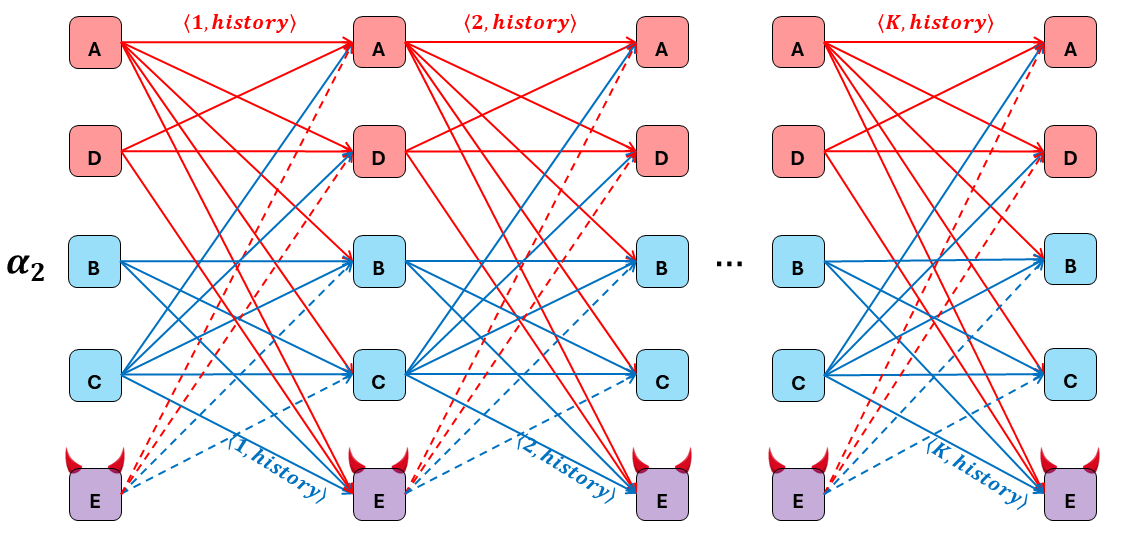}
    
    \includegraphics[width=0.9\textwidth]{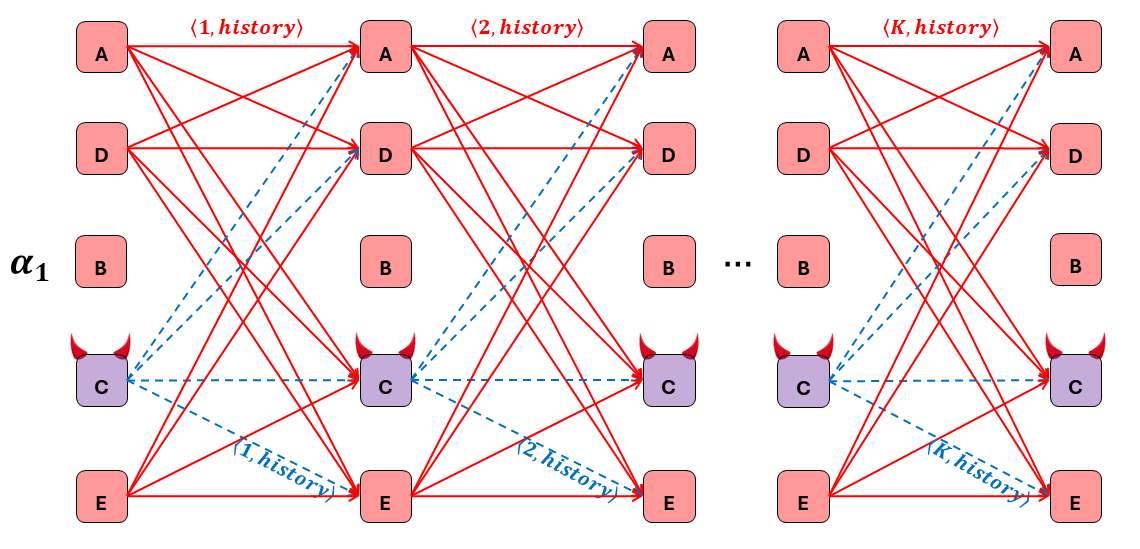}
    
    \caption{Executions $\alpha_0$ (top), $\alpha_2$ (middle), $\alpha_1$ (bottom). Processes colored blue are initialized with input 0, red with input 1 and purple are Byzantine. 
    Note that in $\alpha_2$, messages from processes in group $B$ to processes in groups $A,D$ are
    delayed until after round $K$, and the same is true for messages from processes in group $D$ to
    processes in groups $B,C$.}
    \label{fig:combined}

\end{figure}

The construction uses a $(K+1) \times 5f$ array of elements from $\mathcal{R}$,
denoted $\rho$, where $\mathcal{R}$ is the set from which the random numbers
are chosen.
For each $r$ and $i$, $\rho[r,i]$ contains the random choice made by
process $p_i$ in round $r$ of certain executions.

\paragraph*{Basis step of the construction:}
Define execution $\alpha_2^0[\rho]$, in which processes wake up and send
their round 1 messages, as follows:
\begin{itemize}
\item Start with $C_2$.
\item $ABCDE$ wake up.
\item Each (correct) $p_i \in ABCD$ sends $(1,C_2[i],\rho[1,i])$ 
      to $ABCDE$.
\item Each (faulty) $p_i \in E$ sends $(1,C_0[i],\rho[1,i])$ to $BC$
      and $(1,C_1[i],$ $\rho[1,i])$ to $AD$.  That is, $p_i$
      pretends to $BC$ that its input is 0
      and pretends to $AD$ that its input is 1.
\end{itemize}

Next we construct execution $\alpha_0^0[\rho]$ as follows:
\begin{itemize}
\item Start with $C_0$.
\item $ABCE$ wake up.  (Processes in $D$ are late.)
\item Each (correct) $p_i \in BCE$ sends $(1,C_0[i],\rho[1,i])$ to $ABCDE$.
\item Each (faulty) $p_i \in A$ sends $(1,C_2[i],\rho[1,i])$ to $ABCDE$.
      That is, $p_i$ pretends that it is correct in $\alpha_2^0[\rho]$.
\end{itemize}

Similarly, we construct execution $\alpha_1^0[\rho]$ as follows:
\begin{itemize}
\item Start with $C_1$.
\item $ACDE$ wake up.  (Processes in $B$ are late.)
\item Each (correct) $p_i \in ADE$ sends $(1,C_1[i],\rho[1,i])$ to $ABCDE$.
\item Each (faulty) $p_i \in C$ sends $(1,C_2[i],\rho[1,i])$ to $ABCDE$.
      That is, $p_i$ pretends that it is correct in $\alpha_2^0[\rho]$.
\end{itemize}

\paragraph*{Inductive step of the construction:}
For $r$, $1 \le r \le K$, assume that executions $\alpha_i^{r-1}[\rho]$,
$i = 0,1,2$, have already been defined.

Define execution $\alpha_2^r[\rho]$, in which processes receive round $r$
messages and send round $r+1$ messages, as follows:
\begin{itemize}
\item Start with $\alpha_2^{r-1}[\rho]$.
\item $AD$ receive the round $r$ messages which are in transit at the
      end of $\alpha_2^{r-1}[\rho]$ from $ACDE$.
\item $BC$ receive the round $r$ messages which are in transit at the
      end of $\alpha_2^{r-1}[\rho]$ from $ABCE$.
\item Each (correct) $p_i \in AD$ sends $(r+1,h_i,\rho[r+1,i])$ to $ABCDE$, 
      where $h_i$ is $p_i$'s history at the end of
      $\alpha_2^{r-1}[\rho]$ followed by the round $r$ messages just 
      received from $ACDE$.
\item Each (correct) $p_i \in BC$ sends $(r+1,h_i,\rho[r+1,i])$ to $ABCDE$, 
      where $h_i$ is $p_i$'s history at the end of
      $\alpha_2^{r-1}[\rho]$ followed by the round $r$ messages just
      received from $ABCE$.
\item Each (faulty) $p_i \in E$ sends $(r+1,h_i^0,\rho[r+1,i])$ 
      to $BC$, where $h_i^0$ is the state of $p_i$
      at the end of $\alpha_0^{r-1}[\rho]$ followed by the round $r$ messages
      in transit from $ABCE$ to $p_i$ at the end of $\alpha_0^{r-1}[\rho]$.
\item Each (faulty) $p_i \in E$ sends $(r+1,h_i^1,\rho[r+1,i])$ 
      to $AD$, where $h_i^1$ is the state of $p_i$
      at the end of $\alpha_1^{r-1}[\rho]$ followed by the round $r$ messages
      in transit from $ACDE$ to $p_i$ at the end of $\alpha_1^{r-1}[\rho]$.
\end{itemize}
That is, each faulty process in $E$ pretends to $BC$ that it is correct
in $\alpha_0^r[\rho]$ and pretends to $AD$ that it is correct
in $\alpha_1^r[\rho]$.

\begin{observation}
Every process completes $r$ canonical rounds in $\alpha_2^r[\rho]$.
\end{observation}

Next we construct execution $\alpha_0^r[\rho]$ as follows:
\begin{itemize}
\item Start with $\alpha_0^{r-1}[\rho]$.
\item $ABCE$ receive round $r$ messages which are in transit at the end of
      $\alpha_0^{r-1}[\rho]$ from $ABCE$.
\item Each (correct) $p_i \in BCE$ sends $(r+1,h_i,\rho[r+1,i])$ to
      $ABCDE$, where $h_i$ is $p_i$'s history at the end of $\alpha_0^{r-1}[\rho]$
      followed by the round $r$ messages it just received from $ABCE$.
\item Each (faulty) $p_i \in A$ sends $(r+1,h_i,\rho[r+1,i])$ to $ABCDE$,
      where $h_i$ is $p_i$'s history at the end of $\alpha_2^{r-1}[\rho]$
      followed by the round $r$ messages in transit from $ACDE$
      to $p_i$ at the end of $\alpha_2^{r-1}[\rho]$. 
      That is, $p_i$ pretends that it is correct in $\alpha_2^r[\rho]$.
\end{itemize}

Similarly, we construct execution $\alpha_1^r[\rho]$ as follows:
\begin{itemize}
\item Start with $\alpha_1^{r-1}[\rho]$. 
\item $ACDE$ receive round $r$ messages which are in transit at the
      end of $\alpha_1^{r-1}[\rho]$ from $ACDE$.
\item Each (correct) $p_i \in ADE$ sends $(r+1,h_i,\rho[r+1,i])$ to
      $ABCDE$, where $h_i$ is $p_i$'s history at the end of
      $\alpha_1^{r-1}[\rho]$ followed by the round $r$ messages it just
      received from $ACDE$.
\item Each (faulty) $p_i \in C$ sends $(r+1,h_i,\rho[r+1,i])$ to
      $ABCDE$, where $h_i$ is $p_i$'s history at the end of
      $\alpha_2^{r-1}[\rho]$ followed by the round $r$ messages in
      transit from $ABCE$ to $p_i$ at the end of $\alpha_2^{r-1}[\rho]$.
      That is, $p_i$ pretends that it is correct in $\alpha_2^r[\rho]$.
\end{itemize}

The next lemma is proved by induction on the rounds. 

\begin{restatable}{lemma}{BCsimilarity}
\label{lem:BC-similarity}
For all $r$, $0 \le r \le K$,
$\alpha_2^r[\rho] \similar{BC} \alpha_0^r[\rho]$ and the same set of messages
are in transit from $ABCE$ to $BC$ at the end of both executions.
\end{restatable}

\begin{proof} 
By induction on $r$.  Note that $BC$ are correct in both executions.

\noindent {\it Basis:}  $r = 0$.
Each $p_i \in BC$ has the same input in $C_2$ as in
$C_0$.  In both executions, all $p_i$ does is wake up and send a round 1
message with the same content.  Thus $\alpha_2^0[\rho] \similar{BC}
\alpha_0^0[\rho]$.  

We just argued that processes in $BC$ send the same
messages in the two executions.  Each $p_i \in AE$ sends the same
message in both executions by definition: 
If $p_i \in A$, then its faulty behavior in $\alpha_0^0[\rho]$
is to send the same message that it sends in $\alpha_2^0[\rho]$
where it is correct.
If $p_i \in E$, then its faulty behavior in $\alpha_2^0[\rho]$
is to send the same message that it sends in $\alpha_0^0[\rho]$
where it is correct.

\noindent {\it Inductive step:} $r \ge 1$.
Assume the lemma is true for $r-1$.

By the inductive hypothesis, each $p_i \in BC$ starts its round $r$
with the same history in both $\alpha_2^r[\rho]$ and $\alpha_0^r[\rho]$.
By the inductive hypothesis, the same set of round $r$ messages are
in transit to $p_i$ from $ABCE$ at the beginning of $p_i$'s round $r$
in both executions.
Since by construction, $p_i$ receives round $r$ messages from $ABCE$
during its round $r$ in both executions, it has the same history
at the end of its round $r$ in both executions.
Thus $\alpha_2^r[\rho] \similar{BC} \alpha_0^r[\rho]$.

The argument just made also implies that each process in $BC$
sends the same round $r+1$ messages to $BC$ in the two executions.
Each $p_i \in A$ sends the same round $r+1$ messages to $BC$ in the two
executions because its faulty behavior in $\alpha_0^r[\rho]$ is to send
the same message that it sends in $\alpha_2^r[\rho]$ where it is correct. 
Each $p_i \in E$ sends the same round $r+1$ messages to $BC$ in the two
executions because its faulty behavior in $\alpha_2^r[\rho]$ is to send
to $BC$ the same message that it sends in $\alpha_0^r[\rho]$ where it
is correct.
\end{proof}
The next lemma is proved analogously to the previous, replacing
$\alpha_0$ with $\alpha_1$,
$BC$ with $AD$, 
$ABCE$ with $ACDE$,
$C_0$ with $C_1$,
``$p_i \in AE$'' with ``$p_i \in CE$'' and
``$p_i \in A$'' with ``$p_i \in C$''.

\begin{restatable}{lemma}{ADsimilarity}
\label{lem:AD-similarity}
For all $r$, $0 \le r \le K$,
$\alpha_2^r[\rho] \similar{AD} \alpha_1^r[\rho]$ and the same set of messages
are in transit from $ACDE$ to $AD$ at the end of both executions.
\end{restatable}

\begin{proof} 
By induction on $r$.  Note that $AD$ are correct in both executions.

\noindent {\it Basis:}  $r = 0$.
Each $p_i \in AD$ has the same input in $C_2$ as in
$C_1$.  In both executions, all $p_i$ does is wake up and send a round 1
message with the same content.  Thus $\alpha_2^0[\rho] \similar{AD}
\alpha_1^0[\rho]$.  

We just argued that processes in $AD$ send the same
messages in the two executions.  Each $p_i \in CE$ sends the same
message in both executions by definition: 
If $p_i \in C$, then its faulty behavior in $\alpha_1^0[\rho]$ is to
send the same message that it sends in $\alpha_2^0[\rho]$ where it is 
correct.
If $p_i \in E$, then its faulty behavior in $\alpha_2^0[\rho]$ is to
send the same message that it sends in $\alpha_1^0[\rho]$ where it is
correct.

\noindent {\it Inductive step:}  $r \ge 1$.
Assume the lemma is true for $r-1$.

By the inductive hypothesis, each $p_i \in AD$ starts its round $r$
with the same history in both $\alpha_2^r[\rho]$ and $\alpha_1^r[\rho]$.
By the inductive hypothesis, the same set of round $r$ messages are
in transit to $p_i$ from $ACDE$ at the beginning of $p_i$'s round $r$
in both executions.
Since by construction, $p_i$ receives round $r$ messages from $ACDE$
during its round $r$ in both executions, it has the same history
at the end of its round $r$ in both executions.
Thus $\alpha_2^r[\rho] \similar{AD} \alpha_1^r[\rho]$.

The argument just made also implies that each process in $AD$
sends the same round $r+1$ messages to $AD$ in the two executions.
Each $p_i \in C$ sends the same round $r+1$ messages to $AD$ in the two
executions because its faulty behavior in $\alpha_1^r[\rho]$ is to send the
same message that it sends in $\alpha_2^r[\rho]$ where it is correct.
Each $p_i \in E$ sends the same round $r+1$ messages to $AD$ in the two
executions because its faulty behavior in $\alpha_2^r[\rho]$ is to send
to $AD$ the same message that it sends in $\alpha_1^r[\rho]$ where it
is correct.
\end{proof}

Let $E_2^K = \{\alpha_2^K[\rho] | \rho$ is any $(K+1) \times 5f$ array
of elements from $\mathcal{R}\}$.

\begin{lemma}
\label{lem:decide by K rounds}
There exists an execution in $E_2^K$ in which all correct processes decide.
\end{lemma}

\begin{proof}
The key point is to show that $E_2^K$ consists of all the $K$-round executions, starting from $C_2$, that can be generated by a fixed adversary $\mathcal{A}$.
The construction of $\alpha_2^K[\rho]$ just given for a fixed collection of random
numbers $\rho$ describes the function $\mathcal{A}$ specifying the order of
process steps and message delays up to the point when all processes have
completed $K$ rounds.
(The execution can be extended infinitely in any arbitrary way as long as
every correct process takes an infinite number of steps and all messages 
to correct processes are eventually
received, but we only care about the execution through $K$ rounds.)
Thus $E_2^K$ consists of all $K$-round executions that have the same scheduling as
that described for $\mathcal{A}$ but with all the possible values for the 
collection $\rho$ of random numbers used by the processes.  

Since the expected round complexity $EX[T]$ is assumed to be at most $K$, we have:
\[ 
K \ge EX[T] 
  \ge EX[\mathcal{A},C_2,T] 
  = \sum_{x\in \mathbb{N}} x \cdot \Pr[\mathcal{A},C_2,T = x] ~.
\]

Thus there is at least one value of $x \le K$ such that $\Pr[\mathcal{A},C_2,T = x]$ is positive.  Let $\rho$ be such that $exec(\mathcal{A},C_2,\rho)$ has round complexity $x$.  The execution $exec(\mathcal{A},C_2,\rho)$ is the same as $\alpha_2^K[\rho]$, and in this execution all correct processes decide by round $K$.
\end{proof}

Let $\alpha_2^K[\rho]$ be an execution whose existence is
guaranteed by Lemma~\ref{lem:decide by K rounds}.

By Lemma~\ref{lem:BC-similarity} with $r = K$, $\alpha_2^K[\rho]
\similar{BC} \alpha_0^K[\rho]$ and thus every process $p_i$ in $BC$
decides by the end of $\alpha_0^K[\rho]$.  By the validity condition,
$p_i$ decides 0 in $\alpha_0^K[\rho]$, and thus it also decides 0 in
$\alpha_2^K[\rho]$.
Similarly, by Lemma~\ref{lem:AD-similarity} with $r = K$, $\alpha_2^K[\rho]
\similar{AD} \alpha_1^K[\rho]$ and thus every process $p_i$ in $AD$
decides by the end of $\alpha_1^K[\rho]$.  By the validity condition,
$p_i$ decides 1 in $\alpha_1^K[\rho]$, and thus it also decides 1 in
$\alpha_2^K[\rho]$.

But $\alpha_2^K[\rho]$ violates the agreement condition, a contradiction.
\end{proof}

We have the following corollary for deterministic algorithms:

\begin{corollary}
\label{cor:unbounded canonical rounds - det}
For any deterministic canonical-round algorithm that solves the nontrivial convergence problem with $n \le 5f$,
the worst-case number of rounds until all correct processes decide is greater than $K$, for any integer $K \in \mathbb{N}$. 
\end{corollary}

\begin{proof}
The proof is essentially the same as that of Theorem~\ref{theorem:unbounded canonical rounds} without the need for $\rho$ (the random choices). 
The key difference is that finding an execution in $E_2^K$ in which all correct processes decide (cf.~Lemma~\ref{lem:decide by K rounds}) is much simpler.
$E_2^K$ contains a single execution and by the
contradiction assumption that the worst-case round complexity is at most $K$, every correct process decides in it.
\end{proof}

\subsection{Impossibility with Communication-Closed Canonical Rounds}
\label{subsec:lb-cc}

We model algorithms with communication-closed canonical rounds following~\cite{Coan1988}: processes proceed in canonical rounds, but round $r$ messages that arrive in a later round than $r$ are discarded. 
As before, processes keep sending messages after they decide, and never halt.

Theorem~\ref{theorem:unbounded canonical rounds} does not rule out the possibility that the expected number of rounds for each adversary is finite, but there is no upper bound that holds for all adversaries.
Under the additional constraint of communication-closed behavior, we can prove something stronger, that there is an adversary that can prevent decision in all its executions.

\begin{restatable}{theorem}{CCimpossibility}
\label{theorem:communication closed impossibility}
For any randomized communication-closed canonical-round algorithm that solves the nontrivial convergence
problem with $n \le 5f$,
there exists an adversary $\mathcal{A}$ and initial configuration $I$ such that 
$\Pr[\mathcal{A},I,P] = 0$, where $P$ is the predicate that
all the correct processes eventually decide.
\end{restatable}

The proof is very similar to that of 
Theorem~\ref{theorem:unbounded canonical rounds}
and also works by contradiction.
We define a similar adversary $\mathcal{A}$ to the one in the previous proof.  
The key differences are:
\begin{itemize}
    \item The executions constructed are extended to infinitely many rounds, not just $K$.
    \item The processes that were late to wake up in the previous proof wake up initially but have their messages delivered one round late.
    \item All the messages that are delayed until after round $K$ in the original executions are delivered one round late.
\end{itemize}
The complete proof is given next, with the differences marked in red.

\begin{proof}  Assume towards a contradiction that there exists a 
\change{communication-closed} canonical-round algorithm for nontrivial convergence with $n \le 5f$ 
\change{such that for every adversary $\mathcal{A}$ and initial configuration $I$, $\Pr[\mathcal{A},I,P] > 0$, where $P$ is the predicate that all correct processes eventually decide.}
For convenience, denote the specific values $x_0$, $x_1$, $d_0$, and $d_1$
in the definition of nontrivial convergence by $0$, $1$, $0$, and $1$
respectively.

For simplicity, we assume $n = 5f$, and divide the processes into five
disjoint sets of $f$ processes each:  $A$, $B$, $C$, $D$, and $E$.
We abuse notation and, for any two sets $X$ and $Y$, denote $X \cup Y$
by $XY$.

We consider the following initial configurations of the algorithm:
\begin{itemize}
\item Denote by $C_0$ the initial configuration in which processes in
$BCDE$ are correct and have input 0, while processes in $A$ are Byzantine.
\item Denote by $C_1$ the initial configuration in which processes in
$ABDE$ are correct and have input 1, while processes in $C$ are Byzantine.
\item Denote by $C_2$ the initial configuration in which processes in
$ABCD$ are correct, those in $BC$ have input 0 and those in $AD$ have input 1, 
while processes in $E$ are Byzantine.
\end{itemize}
For each $C_j$, let $C_j[i]$ denote $p_i$'s input in $C_j$.

We next define a set of executions of the algorithm that begin with $C_2$.
For each of these executions, we define a related execution that
begins with $C_0$ and another one that begins with $C_1$.
All these executions are recursively constructed round by round,
with the behavior of the faulty processes in each one being dependent
on the behavior of some correct processes in others.
(See Figures~\ref{figure:three executions} and~\change{\ref{fig:communication_closed_combined}}.)

The construction uses a \change{two-dimensional} array of elements from $\mathcal{R}$,
denoted $\rho$, where $\mathcal{R}$ is the set from which the random numbers are chosen.
For each $r$ and $i$, $\rho[r,i]$ contains the random choice made by process $p_i$ in round $r$ of certain executions.

\begin{figure}

\centering
\includegraphics[width=0.9\textwidth]{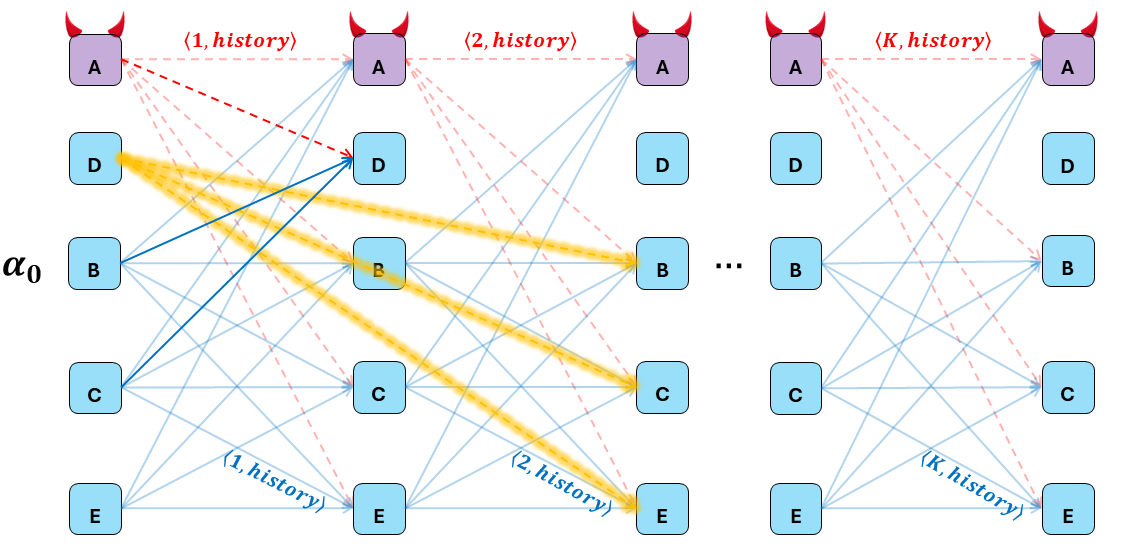}
\vspace{0.5cm} 
    
\includegraphics[width=0.9\textwidth]{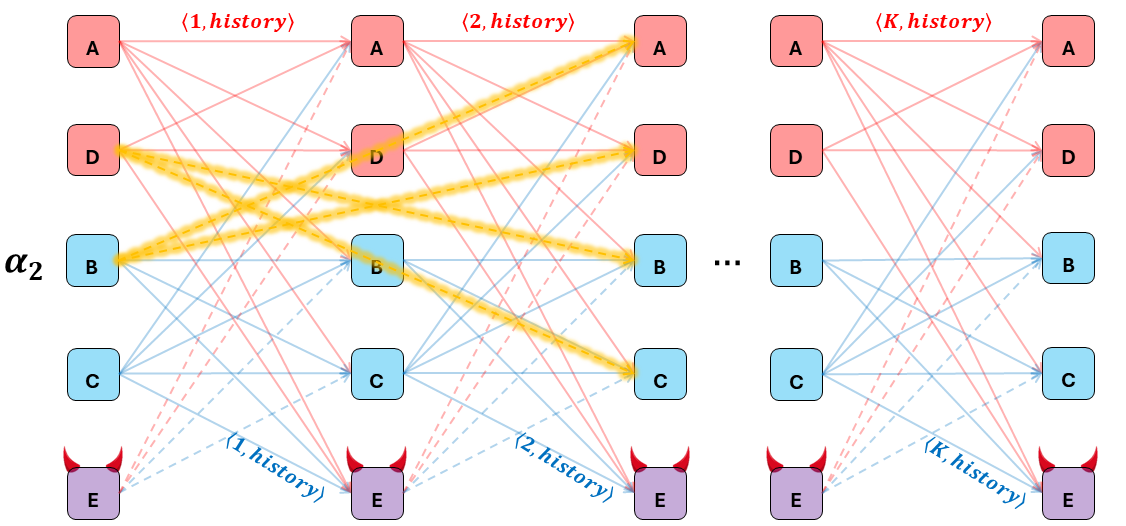}
\vspace{0.5cm} 
    
\includegraphics[width=0.9\textwidth]{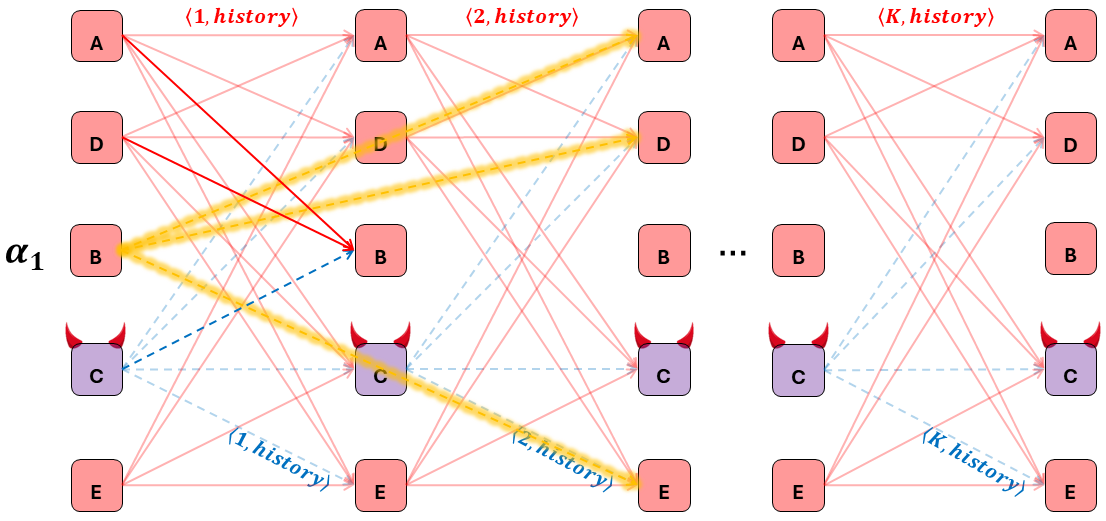}
    
\caption{Illustration of the executions used in Theorem~\ref{theorem:communication closed impossibility}.
In each execution, delayed messages are colored yellow.}
\label{fig:communication_closed_combined}

\end{figure}

\paragraph*{Basis step of the construction:}

Define execution $\alpha_2^0[\rho]$, in which processes wake up and send
their round 1 messages, as follows:
\begin{itemize}
\item Start with $C_2$.
\item $ABCDE$ wake up.
\item Each (correct) $p_i \in ABCD$ sends $(1,C_2[i],\rho[1,i])$ 
      to $ABCDE$.
\item Each (faulty) $p_i \in E$ sends $(1,C_0[i],\rho[1,i])$ to $BC$
      and $(1,C_1[i],\rho[1,i])$ to $AD$.  That is, $p_i$
      pretends to $BC$ that its input is 0
      and pretends to $AD$ that its input is 1.
\end{itemize}

Next we construct execution $\alpha_0^0[\rho]$ as follows:
\begin{itemize}
\item Start with $C_0$.
\item $ABC\change{D}E$ wake up.
\item Each (correct) $p_i \in BC\change{D}E$ sends $(1,C_0[i],\rho[1,i])$ to $ABCDE$.
\item Each (faulty) $p_i \in A$ sends $(1,C_2[i],\rho[1,i])$ to $ABCDE$.
      That is, $p_i$ pretends that it is correct in $\alpha_2^0[\rho]$.
\end{itemize}

Similarly, we construct execution $\alpha_1^0[\rho]$ as follows:
\begin{itemize}
\item Start with $C_1$.
\item $A\change{B}CDE$ wake up.  
\item Each (correct) $p_i \in A\change{B}DE$ sends $(1,C_1[i],\rho[1,i])$ to $ABCDE$.
\item Each (faulty) $p_i \in C$ sends $(1,C_2[i],\rho[1,i])$ to $ABCDE$.
      That is, $p_i$ pretends that it is correct in $\alpha_2^0[\rho]$.
\end{itemize}

\paragraph*{Inductive step of the construction:}

For \change{$r \ge 1$}, assume that executions $\alpha_i^{r-1}[\rho]$,
$i = 0,1,2$, have already been defined.

Define execution $\alpha_2^r[\rho]$, in which processes receive round $r$
messages and send round $r+1$ messages, as follows:
\begin{itemize}
\item Start with $\alpha_2^{r-1}[\rho]$.
\item \change{If $r > 1$, then $ABCDE$ receive all round $r-1$ messages which are in transit
      at the end of $\alpha_2^{r-1}[\rho]$; these messages are ignored.}
\item $AD$ receive the round $r$ messages which are in transit at the
      end of $\alpha_2^{r-1}[\rho]$ from $ACDE$.
\item $BC$ receive the round $r$ messages which are in transit at the
      end of $\alpha_2^{r-1}[\rho]$ from $ABCE$.
\item Each (correct) $p_i \in AD$ sends $(r+1,h_i,\rho[r+1,i])$ to $ABCDE$, 
      where $h_i$ is $p_i$'s history at the end of
      $\alpha_2^{r-1}[\rho]$ followed by the round $r$ messages just 
      received from $ACDE$.
\item Each (correct) $p_i \in BC$ sends $(r+1,h_i,\rho[r+1,i])$ to $ABCDE$, 
      where $h_i$ is $p_i$'s history at the end of
      $\alpha_2^{r-1}[\rho]$ followed by the round $r$ messages just
      received from $ABCE$.
\item Each (faulty) $p_i \in E$ sends $(r+1,h_i^0,\rho[r+1,i])$ 
      to $BC$, where $h_i^0$ is the state of $p_i$
      at the end of $\alpha_0^{r-1}[\rho]$ followed by the round $r$ messages
      in transit from $ABCE$ to $p_i$ at the end of $\alpha_0^{r-1}[\rho]$.
\item Each (faulty) $p_i \in E$ sends $(r+1,h_i^1,\rho[r+1,i])$ 
      to $AD$, where $h_i^1$ is the state of $p_i$
      at the end of $\alpha_1^{r-1}[\rho]$ followed by the round $r$ messages
      in transit from $ACDE$ to $p_i$ at the end of $\alpha_1^{r-1}[\rho]$.
\end{itemize}
That is, each faulty process in $E$ pretends to $BC$ that it is correct
in $\alpha_0^r[\rho]$ and pretends to $AD$ that it is correct
in $\alpha_1^r[\rho]$.

\noindent
{\bf Observation:}
Every process completes $r$ canonical rounds in $\alpha_2^r[\rho]$.

Next we construct execution $\alpha_0^r[\rho]$ as follows:
\begin{itemize}
\item Start with $\alpha_0^{r-1}[\rho]$.
\item \change{If $r > 1$, then $ABCDE$ receive all round $r-1$ messages which are in transit
      at the end of $\alpha_2^{r-1}[\rho]$; these messages are ignored.}
\item $ABCE$ receive round $r$ messages which are in transit at the end of
      $\alpha_0^{r-1}[\rho]$ from $ABCE$.
\item Each (correct) $p_i \in BCE$ sends $(r+1,h_i,\rho[r+1,i])$ to
      $ABCDE$, where $h_i$ is $p_i$'s history at the end of $\alpha_0^{r-1}[\rho]$
      followed by the round $r$ messages it just received from $ABCE$.
\item Each (faulty) $p_i \in A$ sends $(r+1,h_i,\rho[r+1,i])$ to $ABCDE$,
      where $h_i$ is $p_i$'s history at the end of $\alpha_2^{r-1}[\rho]$
      followed by the round $r$ messages in transit from $ACDE$
      to $p_i$ at the end of $\alpha_2^{r-1}[\rho]$. 
      That is, $p_i$ pretends that it is correct in $\alpha_2^r[\rho]$.
\end{itemize}

Similarly, we construct execution $\alpha_1^r[\rho]$ as follows:
\begin{itemize}
\item Start with $\alpha_1^{r-1}[\rho]$. 
\item \change{If $r > 1$, then $ABCDE$ receive all round $r-1$ messages which are in transit
      at the end of $\alpha_2^{r-1}[\rho]$; these messages are ignored.}
\item $ACDE$ receive round $r$ messages which are in transit at the
      end of $\alpha_1^{r-1}[\rho]$ from $ACDE$.
\item Each (correct) $p_i \in ADE$ sends $(r+1,h_i,\rho[r+1,i])$ to
      $ABCDE$, where $h_i$ is $p_i$'s history at the end of
      $\alpha_1^{r-1}[\rho]$ followed by the round $r$ messages it just
      received from $ACDE$.
\item Each (faulty) $p_i \in C$ sends $(r+1,h_i,\rho[r+1,i])$ to
      $ABCDE$, where $h_i$ is $p_i$'s history at the end of
      $\alpha_2^{r-1}[\rho]$ followed by the round $r$ messages in
      transit from $ABCE$ to $p_i$ at the end of $\alpha_2^{r-1}[\rho]$.
      That is, $p_i$ pretends that it is correct in $\alpha_2^r[\rho]$.
\end{itemize}

\begin{lemma}
\label{lem:BC-similarity-CC}
For all \change{$r \ge 0$},
$\alpha_2^r[\rho] \similar{BC} \alpha_0^r[\rho]$ and the same set of messages
are in transit from $ABCE$ to $BC$ at the end of both executions.
\end{lemma}

\begin{proof} 
By induction on $r$.  Note that $BC$ are correct in both executions.

\noindent {\it Basis:}  $r = 0$.
Each $p_i \in BC$ has the same input in $C_2$ as in
$C_0$.  In both executions, all $p_i$ does is wake up and send a round 1
message with the same content.  Thus $\alpha_2^0[\rho] \similar{BC}
\alpha_0^0[\rho]$.  

We just argued that processes in $BC$ send the same
messages in the two executions.  Each $p_i \in AE$ sends the same
message in both executions by definition: 
If $p_i \in A$, then its faulty behavior in $\alpha_0^0[\rho]$
is to send the same message that it sends in $\alpha_2^0[\rho]$
where it is correct.
If $p_i \in E$, then its faulty behavior in $\alpha_2^0[\rho]$
is to send the same message that it sends in $\alpha_0^0[\rho]$
where it is correct.

\noindent {\it Inductive step:} $r \ge 1$.
Assume the lemma is true for $r-1$.

By the inductive hypothesis, each $p_i \in BC$ starts its round $r$
with the same history in both $\alpha_2^r[\rho]$ and $\alpha_0^r[\rho]$.
By the inductive hypothesis, the same set of round $r$ messages are
in transit to $p_i$ from $ABCE$ at the beginning of $p_i$'s round $r$
in both executions.
Since by construction, $p_i$ receives round $r$ messages from $ABCE$
during its round $r$ in both executions, it has the same history
at the end of its round $r$ in both executions.
Thus $\alpha_2^r[\rho] \similar{BC} \alpha_0^r[\rho]$.

The argument just made also implies that each process in $BC$
sends the same round $r+1$ messages to $BC$ in the two executions.
Each $p_i \in A$ sends the same round $r+1$ messages to $BC$ in the two
executions because its faulty behavior in $\alpha_0^r[\rho]$ is to send
the same message that it sends in $\alpha_2^r[\rho]$ where it is correct. 
Each $p_i \in E$ sends the same round $r+1$ messages to $BC$ in the two
executions because its faulty behavior in $\alpha_2^r[\rho]$ is to send
to $BC$ the same message that it sends in $\alpha_0^r[\rho]$ where it
is correct.
\end{proof}

The next lemma is proved analogously to the previous, replacing
$\alpha_0$ with $\alpha_1$,
$BC$ with $AD$, 
$ABCE$ with $ACDE$,
$C_0$ with $C_1$,
``$p_i \in AE$'' with ``$p_i \in CE$'' and
``$p_i \in A$'' with ``$p_i \in C$''.

\begin{lemma}
\label{lem:AD-similarity-CC}
For all \change{$r \ge 0$},
$\alpha_2^r[\rho] \similar{AD} \alpha_1^r[\rho]$ and the same set of messages
are in transit from $ACDE$ to $AD$ at the end of both executions.
\end{lemma}

\begin{proof} 
By induction on $r$.  Note that $AD$ are correct in both executions.

\noindent {\it Basis:}  $r = 0$.
Each $p_i \in AD$ has the same input in $C_2$ as in
$C_1$.  In both executions, all $p_i$ does is wake up and send a round 1
message with the same content.  Thus $\alpha_2^0[\rho] \similar{AD}
\alpha_1^0[\rho]$.  

We just argued that processes in $AD$ send the same
messages in the two executions.  Each $p_i \in CE$ sends the same
message in both executions by definition: 
If $p_i \in C$, then its faulty behavior in $\alpha_1^0[\rho]$ is to
send the same message that it sends in $\alpha_2^0[\rho]$ where it is 
correct.
If $p_i \in E$, then its faulty behavior in $\alpha_2^0[\rho]$ is to
send the same message that it sends in $\alpha_1^0[\rho]$ where it is
correct.

\noindent {\it Inductive step:}  $r \ge 1$.
Assume the lemma is true for $r-1$.

By the inductive hypothesis, each $p_i \in AD$ starts its round $r$
with the same history in both $\alpha_2^r[\rho]$ and $\alpha_1^r[\rho]$.
By the inductive hypothesis, the same set of round $r$ messages are
in transit to $p_i$ from $ACDE$ at the beginning of $p_i$'s round $r$
in both executions.
Since by construction, $p_i$ receives round $r$ messages from $ACDE$
during its round $r$ in both executions, it has the same history
at the end of its round $r$ in both executions.
Thus $\alpha_2^r[\rho] \similar{AD} \alpha_1^r[\rho]$.

The argument just made also implies that each process in $AD$
sends the same round $r+1$ messages to $AD$ in the two executions.
Each $p_i \in C$ sends the same round $r+1$ messages to $AD$ in the two
executions because its faulty behavior in $\alpha_1^r[\rho]$ is to send the
same message that it sends in $\alpha_2^r[\rho]$ where it is correct.
Each $p_i \in E$ sends the same round $r+1$ messages to $AD$ in the two
executions because its faulty behavior in $\alpha_2^r[\rho]$ is to send
to $AD$ the same message that it sends in $\alpha_1^r[\rho]$ where it
is correct.
\end{proof}

\change{
The construction of $\alpha_2^r[\rho]$ just given for a fixed $\rho$ and all $r$ describes a function $\mathcal{A}$ that, starting from $C_2$, specifies the order of process steps, message delays, and the contents of messages from faulty processes. Since every process takes infinitely many steps as $r$ goes to infinity and every message is received, $\mathcal{A}$ is an adversary.
By the contradiction assumption, $\Pr[\mathcal{A},C_2,P] > 0$, which implies there exists $\rho$ such that in $exec(\mathcal{A},C_2,\rho)$, $P$ is true (all correct processes eventually decide).
Let $\rho$ be one such collection of random numbers and let $K$ be the round by which the decisions are made.
}

By Lemma~\change{\ref{lem:BC-similarity-CC}} with $r = K$, $\alpha_2^K[\rho]
\similar{BC} \alpha_0^K[\rho]$ and thus every process $p_i$ in $BC$
decides by the end of $\alpha_0^K[\rho]$.  By the validity condition,
$p_i$ decides 0 in $\alpha_0^K[\rho]$, and thus it also decides 0 in
$\alpha_2^K[\rho]$.

Similarly, by Lemma~\change{\ref{lem:AD-similarity-CC}} with $r = K$, $\alpha_2^K[\rho]
\similar{AD} \alpha_1^K[\rho]$ and thus every process $p_i$ in $AD$
decides by the end of $\alpha_1^K[\rho]$.  By the validity condition,
$p_i$ decides 1 in $\alpha_1^K[\rho]$, and thus it also decides 1 in
$\alpha_2^K[\rho]$.

But $\alpha_2^K[\rho]$ violates the agreement condition, a contradiction.
\end{proof}

We have the following corollary for deterministic algorithms:

\begin{corollary}
\label{cor:imposs communication closed - det}
It is impossible for any deterministic communication-closed canonical-round algorithm to solve the nontrivial convergence
problem with $n \le 5f$.
\end{corollary}

\begin{proof}
Assume in contradiction there is such an algorithm.
Consider the adversary $\mathcal{A}$ and initial configuration $I$ from the statement of Theorem~\ref{theorem:communication closed impossibility}.  
When applied to the deterministic algorithm, there is only one resulting execution and it does not satisfy the predicate that all the correct processes eventually decide.
\end{proof}

\subsection{Immediate Consequences}
\label{subsec:immediate}

We now examine the consequences of
Theorems~\ref{theorem:unbounded canonical rounds} and~\ref{theorem:communication closed impossibility} and their deterministic corollaries for several fundamental tasks in asynchronous Byzantine-tolerant systems.
We focus on deterministic solutions for all the tasks considered, except for consensus which has only randomized solutions.

Recall from Section~\ref{subsec:canonical-round-defs} that {\bf consensus}, {\bf $\epsilon$-approximate agreement} (where the input set contains two values more than $\epsilon$ apart), {\bf approximate agreement on a graph} (with two vertices at distance 2), and {\bf $R$-connected consensus} (which generalizes {\bf crusader agreement} and {\bf gradecast}) are all examples of nontrivial convergence problems.

Thus Theorem~\ref{theorem:unbounded canonical rounds} implies that for each of these problems, the expected round complexity of any \emph{randomized} canonical-round algorithm with $n \le 5f$ is unbounded.  Furthermore, by Theorem~\ref{theorem:communication closed impossibility} such an algorithm cannot be communication closed.

In addition, Corollary~\ref{cor:unbounded canonical rounds - det} implies that any \emph{deterministic} canonical-round algorithm for each of these problems with $n \le 5f$ has an execution in which some correct process does not decide by round $K$, for any $K \in \mathbb{N}$.
Also, by Corollary~\ref{cor:imposs communication closed - det}, such an algorithm cannot be communication closed.

\section{Canonical-Round Lower Bounds via Reductions}
\label{sec:reductions}

In this section we use reductions to apply our lower bounds from Section~\ref{sec:lb} beyond the realm of nontrivial convergence problems.
We show that crusader agreement (i.e., connected consensus with $R = 1$) can be used in a reduction to show that reliable broadcast is subject to the same limitations.
Crusader agreement can also be used to solve a problem called \emph{gather} with no additional communication, and thus the limitations on crusader agreement also apply to gather.

\subsection{Reliable and Consistent Broadcast}
\label{app:bcast reductions}

\emph{Reliable broadcast}~\cite{Bracha1987} is defined 
with one of the $n$ processes, $s$,
designated as the \emph{sender}.
The sender has an input value $v$, 
and it calls \emph{r-broadcast}$(v,s)$, where the argument $s$ indicates that $s$ is the sender in this instantiation. 
Processes other than $p$ call \emph{r-broadcast}$(-,s)$, where the argument $-$ indicates that the invoker is not the sender in this instantiation.
Processes
may terminate with \emph{r-accept}($w$,$s$),
with the following properties:
\begin{description}
\item[Agreement:] All correct processes that accept a value from sender $s$, accept the same value.
\item[Validity:] If the sender $s$ is correct then eventually all correct processes accept $s$’s input.
\item[Totality (relay):] If some correct process accepts a value from sender $s$ 
then eventually all correct processes accept a value from sender $s$.
\end{description}

A weaker variant, \textit{consistent broadcast}~\cite{CachinKPS2001}, 
is very similar
but abandons the totality property; it requires only that correct processes do not accept conflicting values from the same sender, and that if the sender is correct then eventually all correct processes accept its message.
The interface to consistent broadcast is denoted c-broadcast and c-accept. 

We use a reduction to show that consistent broadcast has no bounded-round canonical round algorithm and no communication-closed canonical round algorithm when $n \le 5f$.
Since every algorithm for reliable broadcast is also a solution to consistent broadcast, the impossibility results for consistent broadcast immediately extend to reliable broadcast.
Consider Algorithm~\ref{alg:RB to CA} for crusader agreement, assuming $n > 4f$, which uses $n$ concurrent instantiations of consistent broadcast. 
Next we show that this algorithm is correct.

\begin{algorithm}[bt]
  \caption{Crusader agreement using consistent broadcast ($n>4f$); code for process $p_i$}
  \label{alg:RB to CA}
  \small
  \begin{algorithmic}[1]
    \State $W_i \gets \emptyset$ \Comment{$W_i$ is a multiset of values}
    \State c-broadcast($v_i$, $p_i$) \Comment{invoke c-broadcast as sender; $v_i$ is $p_i$'s input}
    \State c-broadcast($-$,$p_j$) for all $j \ne i$ 
    \Comment{invoke $n-1$ c-broadcasts as non-sender}
    \State \textbf{repeat}
        \Statex \quad \textbf{upon} c-accept($v$, $p_j$): 
        $W_i \gets W_i \cup\{v\}$
    \State \textbf{until} $|W_i| = n - f$ 
    \State \textbf{if} $W_i$ contains $|W_i|-f$ copies of $v$ \textbf{then} \textbf{decide} $v$
           {else} \textbf{decide} $\bot$ \textbf{endif}
  \end{algorithmic}
  \label{alg:ca-using-rbcast}
\end{algorithm}

\begin{theorem}
\label{thm:cc-from-cb:proof}
Algorithm~\ref{alg:ca-using-rbcast} solves crusader agreement for $n>4f$.
\end{theorem}

\begin{proof}
To argue \emph{agreement} for crusader agreement, 
assume for contradiction that 
a correct process $p_i$ has $|W_i|-f$ copies of $v$ in $W_i$, 
and a correct process $p_j$ has $|W_j|-f$ copies of $w$ in $W_j$.
Then, since $|W_i|, |W_j| \geq n-f$ and since $n > 4f$, 
$p_i$ has c-accepted $v$ from some process, 
while $p_j$ has c-accepted $w$ from the same process, 
in contradiction to the agreement property of consistent broadcast. 

To argue \emph{validity} for crusader agreement, it is clear that when all correct processes start with $v$, each correct process will c-accept at least $|W_i|-f$ copies of $v$ and thus decide $v$.

To show \emph{termination} for crusader agreement, note that Algorithm~\ref{alg:ca-using-rbcast}
simply waits for the termination of $n-f$ 
out of $n$ concurrent invocations of consistent broadcast.
\end{proof}

Since Algorithm~\ref{alg:ca-using-rbcast} adds no communication beyond that of the $n$ copies of consistent broadcast that run in parallel, if the consistent broadcast is a (communication-closed) canonical round algorithm, then so is Algorithm~\ref{alg:ca-using-rbcast}.
This contradicts the results of Section~\ref{subsec:immediate}, and implies the following corollary, which justifies why Algorithm~\ref{alg:bracha} \cite{Bracha1987} is not in canonical round form:

\begin{corollary}
\label{cor:reliable broadcast}
In the asynchronous model with $n \le 5f$, any canonical round algorithm for reliable or consistent broadcast has an execution in which some correct process does not terminate by round $K$, for any integer $K \geq 1$.
Furthermore, such an algorithm is not communication-closed.
\end{corollary}

\subsection{Gather}

\emph{Gather}'s specification (e.g.,~\cite{SternA2021}) can be viewed as an extension of reliable
broadcast in which \emph{all processes} broadcast their value, and accept values from a large set of processes.
Beyond properties inherited from reliable
broadcast, most notably, that if two correct processes accept a value from another process, it is the same value, gather also ensures that there is a \emph{common core} of $n-f$ values that are accepted by all correct processes. 
In more detail, \emph{gather} is called by process $p_i$ with an input $x_i$ and it returns a set $S_i$ of distinct (process id, value) pairs satisfying the following properties:
\begin{description}
\item[Agreement:] For any $k$, if $p_i$ and $p_j$ are correct and $(k,x) \in S_i$ and $(k,x') \in S_j$, then $x = x'$.
\item[Validity:] For every pair of correct processes $p_i$ and $p_j$, if $(j,x) \in S_i$, then $x=x_j$.
\item[Termination:]  All correct processes eventually return.
\item[Common core:] There is a set $S^C$ of size $n-f$ such that $S^C \subseteq S_i$, for every correct process $p_i$.
\end{description}

We use a reduction from connected consensus with $R = 1$ (i.e., crusader agreement) to gather to show that gather has no bounded-round canonical round algorithm and no communication-closed canonical round algorithm when $n \le 5f$.
We in fact present a connected consensus algorithm for
any $R \ge 1$ that uses subroutines for gather and, if $R \ge 2$,
for approximate agreement.
The algorithm assumes $n > 3f$ and thus relies on subroutines
that also work when $n > 3f$.
See Algorithm~\ref{alg:cc 3f} for the pseudocode.
Here we focus on Lines 1--3 for the $R = 1$ case;
note that this part of the algorithm makes one call to a gather subroutine and does no additional communication.

For completeness, we present in Appendix~\ref{app:gather} 
algorithms for non-binding gather and for binding gather, which use reliable broadcast.

\begin{algorithm}[tb]
\caption{Modular Byzantine-tolerant $R$-connected consensus for arbitrary $R \ge 1$ ($n > 3f$); code for process $p_i$.
Uses 
subroutines G for gather and AA for approximate agreement.}

\label{alg:cc 3f}

\begin{algorithmic}[1]{}
\small
\State $S_i \leftarrow$ G$(x_i)$ \Comment{call gather subroutine; $x_i$ is $p_i$'s input}
       \label{line:cc 3f:first gather}
\State \textbf{if} some $v \in V$ appears $|S_i| - f$ times in $S_i$
       \textbf{then} $v_i \leftarrow v$; $g_i \leftarrow 1$
       \textbf{else} $v_i \leftarrow \bot$; $g_i \leftarrow 0$
       \textbf{endif}
       \label{line:cc 3f:initial estimates}
\State \textbf{if} $R = 1$ 
       \textbf{then} decide $(v_i,g_i)$
       \textbf{endif}
       \Comment{crusader agreement case; return after deciding}
\Statex \Comment{rest of code is for $R \ge 2$ case}
\State $T_i \leftarrow$ G$(v_i)$ 
    \Comment{call gather subroutine again, now with only one non-$\bot$ value}
       \label{line:cc 3f:second gather}
\State \textbf{if} some $v \in V$ appears $f+1$ times in $T_i$
\Statex \quad\quad \textbf{then} $v_i \leftarrow v$
       \textbf{else} $v_i \leftarrow \bot$ 
       \textbf{endif}
       \Comment{choose unique branch}
       \label{line:cc 3f:choose branch}
\State \textbf{if} some $v \in V$ appears $|T_i| - f$ times in $T_i$
\Statex \quad\quad \textbf{then} $r_i \leftarrow R$
       \textbf{else} $r_i \leftarrow 0$
       \textbf{endif}
       \Comment{choose AA input}
       \label{line:cc 3f:choose AA input}
\State $g_i \leftarrow \lceil$AA$(1,r_i)\rceil$
       \Comment{choose grade via AA with $\epsilon = 1$}
       \label{line:cc 3f:call AA}
\State \textbf{if} $g_i = 0$
       \textbf{then} decide $(\bot,0)$
       \textbf{else} decide $(v_i,g_i)$
       \textbf{endif}
       \Comment{will show $v_i \ne \bot$ when $g_i > 0$}
\end{algorithmic}
\end{algorithm}

To prove the correctness of Algorithm~\ref{alg:cc 3f}, we start with three facts about the outputs of any gather primitive. 
Let $S_i$ be the output of a correct process $p_i$ and $S^C$ be the guaranteed common core.

\begin{lemma}
\label{lem:cc 3f:gather result}
If a value $v \in V$ appears at least $|S_i| - f$ times in $S_i$, then 
$v$ appears at least $|S^C| - f$ times in $S^C$.
\end{lemma}

\begin{proof}
For any set $S$ of (process-id, value) pairs and any value $v$, 
let $\#(S,v)$ be the number of pairs in $S$ containing $v$.
By the hypothesis of the lemma, $\#(S_i,v) \geq |S_i| - f$.
Let $T_i = S_i \setminus S^C$ be the subset of $S_i$ that is not in the 
common core $S^C$; then $|T_i| = |S_i| - |S^C|$.
We have as needed:
\[
\#(S^C,v) = \#(S_i,v) - \#(T_i,v) \geq |S_i| - f - (|S_i| - |S^C|) = |S^C| -f .
\qedhere \] 
\end{proof}

Since $|S^C| = n-f$ and $n > 3f$, we also have:

\begin{proposition}
\label{prop:cc 3f:gather core unique}
At most one value can appear $|S^C| - f$ times in $S^C$.
\end{proposition}

\begin{proposition}
\label{prop:cc 3f:gather core valid}
If some $v \in V$ appears $|S_i| - f$ times in $S_i$, 
then $v$ is the input of some correct process. 
\end{proposition}

\begin{proof}
Since $|S_i| \ge n-f$ and $n > 3f$, at least one of the occurrences of $v$ in $S_i$ is for a correct process $p_j$.  By the validity property of gather, $v$ is $p_j$'s input.
\end{proof}

\begin{restatable}{theorem}{thmR}
\label{thm:cc 3f:correct when R is 1}
Algorithm~\ref{alg:cc 3f} solves connected consensus for $R = 1$.
\end{restatable}

\begin{proof}
\emph{Termination} holds by the assumed correctness, including termination, of the G subroutine.

\emph{Validity:}  Suppose correct process $p_i$ decides $(v,g)$.

(i) If $v \ne \bot$, we must show that some correct process has input $v$.
Since $v \ne \bot$, by Line~\ref{line:cc 3f:initial estimates} there are
$|S_i| - f$ copies of $v$ in $S_i$.
Proposition~\ref{prop:cc 3f:gather core valid} gives the result.

(ii) If all correct processes have the same input, we must show $g = 1$.
Suppose all the correct processes have the same input, 
which we know from (i) is $v$.
By the validity property of gather,
all the elements in $S_i$ for correct processes, of which there are
at least $|S_i|-f$, have value $v$.
Thus $p_i$ sets $g_i$ to $1$ in Line~\ref{line:cc 3f:initial estimates}
and $g = 1$.

\emph{Agreement:}  By Line~\ref{line:cc 3f:initial estimates},
the only possible decisions are $(\bot,0)$ and $(v,1)$ where $v \in V$.
Agreement could be violated only if some correct process
$p_i$ decides $(u,1)$ and another correct process $p_j$ decides $(v,1)$
where $u \ne v$.
Since $u$ appears $|S_i| - f$ times in $S_i$,
Lemma~\ref{lem:cc 3f:gather result} implies that
$u$ appears $|S^C| - f$ times in $S^C$, 
and similarly $v$ appears $|S^C| - f$ times in $S^C$.
By Proposition~\ref{prop:cc 3f:gather core unique}, $u$ must equal $v$.
\end{proof}

If the gather subroutine G used in Algorithm~\ref{alg:cc 3f} is a (communication-closed) canonical-round algorithm, then so is Algorithm~\ref{alg:cc 3f}.  
Since Algorithm~\ref{alg:cc 3f} does not add any communication on top of the gather subroutine when $R = 1$,
the discussion in Section~\ref{subsec:immediate} regarding crusader agreement implies:

\begin{corollary}
\label{cor:gather-lb}
In the asynchronous model with $n \le 5f$, any canonical-round algorithm for gather has an execution in which some correct process does not terminate by round $K$, for any integer $K \geq 1$.\\
Furthermore, such an algorithm is not communication-closed.
\end{corollary}

\section{Modular Connected Consensus Algorithm with Optimal Resilience}
\label{app:cc 3f}

Algorithm~\ref{alg:cc 3f} solves connected consensus for any $R \ge 1$ when $n > 3f$.  
After exchanging inputs via the first call to the gather subroutine G, correct processes have at most one non-$\bot$ value.
When $R \ge 2$, process exchange their values again via a second call to G.
Then a process chooses for its branch a value that appears at least $f+1$ times.
Finally a process chooses its grade by running the approximate agreement subroutine AA with input either $R$, if some value appears overwhelmingly often, or 0 otherwise.
Since AA returns a real number between 0 and $R$, the grade is
set to the ceiling of the AA output.
In fact, using a general approximate agreement subroutine is a bit of overkill; see Appendix~\ref{app:2input-AA} for a simplification.

\subsection{Correctness for Arbitrary $R \geq 1$}

In order to prove correctness of Algorithm~\ref{alg:cc 3f} 
for any value of $R$ greater than 1, we need the following two lemmas.
The first one shows that, thanks to the second call to the gather
subroutine, the branch chosen in Line~\ref{line:cc 3f:choose branch}
is unique.
The second one shows that the decisions are well-defined.

\begin{lemma}
\label{lem:cc 3f:unique branch}
If $v \in V$ occurs $f+1$ times in $T_i$ and $u \in V$ appears
$f+1$ times in $T_j$ for not necessarily distinct correct processes
$p_i$ and $p_j$, then $u = v$.
\end{lemma}

\begin{proof}
At least one occurrence of $v$ in $T_i$ is for a correct process $p_k$
and by the validity property for gather, $v$ is $p_k$'s input
to its second call to G on Line~\ref{line:cc 3f:second gather}.
By the code (cf.\ Line~\ref{line:cc 3f:initial estimates}), $v$ occurs
$|S_k| - f$ times in $S_k$, 
and by Lemma~\ref{lem:cc 3f:gather result}, $v$ occurs
$|S^C| - f$ times in $S^C$.

The same argument shows that $u$ occurs $|S^C| - f$ times in $S^C$.

Proposition~\ref{prop:cc 3f:gather core unique} implies that $u = v$.
\end{proof}

\begin{lemma}
\label{lem:cc 3f:well defined decision}
If $R \ge 2$ and correct process $p_i$ decides $(v,g)$, then $g = 0$
if and only if $v = \bot$.
\end{lemma}

\begin{proof}
\emph{Case 1:} $g_i$ is set to 0 in Line~\ref{line:cc 3f:call AA}.
Then by the code the decision is $(\bot,0)$ and so $v = \bot$.

\emph{Case 2:} $g_i$ is set to a positive value in 
Line~\ref{line:cc 3f:call AA}.
By the validity property of approximate agreement, some correct process
$p_j$ calls AA with input $r_j = R$.
Thus there is some value $u \in V$ that appears $|T_j| - f$ times
in $T_j$.
By Lemma~\ref{lem:cc 3f:gather result}, $u$ appears $|T^C| - f$ times in $T^C$.
By the common core property of gather, $u$ also appears $|T^C| - f$
times in $T_i$.
Since $|T^C| - f = n - 2f$ and $n > 3f$, $u$ appears $f+1$ times
in $T_i$.
By Lemma~\ref{lem:cc 3f:unique branch}, $u$ is the only non-$\bot$
value that appears $f+1$ times in $T_i$.
Thus $p_i$ sets $v_i$ to $u$ in Line~\ref{line:cc 3f:choose branch}.
Since $u \ne \bot$ and $u$ is $v$, $v \ne \bot$.
\end{proof}

\begin{theorem}
\label{thm:cc 3f:correct for any R}
Algorithm~\ref{alg:cc 3f} is correct for any value of $R \ge 2$.
\end{theorem}

\begin{proof}
\emph{Termination} follows from the assumed correctness, including
termination, of the G and AA subroutines.

\emph{Validity:}  Suppose correct process $p_i$ decides $(v,g)$.

(i) If $v \ne \bot$, we must show that some correct process has input $v$.
By the code, $v$ is some value that appears $f+1$ times in $T_i$
(cf.\ Line~\ref{line:cc 3f:choose branch}).
At least one of these appearances of $v$ is for a correct process $p_j$.
By the validity property of gather, $p_j$ makes its second call to G
with input $v$, implying that $v$ appears $|S_j| - f$ times in $S_j$
(cf.\ Line~\ref{line:cc 3f:initial estimates}).
By Proposition~\ref{prop:cc 3f:gather core valid},
$v$ is the input of a correct process. 

(ii) If all correct processes have the same input, we must show $g = R$.
Suppose all correct inputs have input $v$.
Then they call G the first time with input $v$.
By the validity property of gather, $v$ appears $|S_j|-f$ times
in $S_j$ for each correct process $p_j$.
It follows that each correct process calls G the second time with
input $v$.
Again by the validity property of gather, $v$ appears $|T_j| - f$
times in $T_j$ for each correct process $p_j$.
Then each $p_j$ sets $v_j$ to $v$ and $r_j$ to $R$ in
Lines~\ref{line:cc 3f:choose branch} and \ref{line:cc 3f:choose AA input}.
By the validity property of approximate agreement, each correct process
gets $R$ as the output of AA and decides $(v,R)$.
Thus the $g$ component in $p_i$'s decision equals $R$.

\emph{Agreement:}  Suppose correct process $p_i$ decides $(u,g)$
and correct process $p_j$ decides $(v,h)$.
By Lemma~\ref{lem:cc 3f:well defined decision}, 
$g = 0$ if and only if $u = \bot$, and $h = 0$ if and only if
$v = \bot$.

Thus if $g$ and $h$ are both 0, we have $(u,g) = (v,h)$.

Suppose exactly one of $g$ and $h$ is 0, say $g$.
By the 1-agreement property of approximate agreement, $h = 1$.
Thus $(u,g) = (\bot,0)$ and $(v,h) = (v,1)$.

Suppose neither $g$ nor $h$ is 0.
By the 1-agreement property of approximate agreement, $|g - h| \le 1$.
Since $u$ appears $f+1$ times in $T_i$ and $v$ appears $f+1$ times in
$T_j$ (cf.\ Line~\ref{line:cc 3f:choose branch}), 
Lemma~\ref{lem:cc 3f:unique branch} implies $u = v$.
\end{proof}

\subsection{Binding and its Preservation}

\emph{Binding} is a property of algorithms, first proposed for crusader agreement~\cite{AbrahamBDY2022} and extended to connected consensus~\cite{AttiyaW2023}, which restricts the power of the adversary to influence later decisions once the first correct process decides.  In more detail, the {\bf binding property for connected consensus} states that for any execution prefix ending when the first correct process decides, there is a value $v \in V$ such that in every extension of $\alpha$, no correct process decides $(u,g)$ for any $u \ne v$.

We define a similar property for gather, which restricts the power of the adversary to influence the common core once the first correct process finishes.  In more detail, the {\bf binding property for gather} states that for any execution prefix ending when the first correct process finishes, there is a set $S^C$ of size $n-f$ such that in every extension of $\alpha$ the output of every correct process contains $S^C$. 

\begin{figure}
    \centering
    \includegraphics[width=0.45\textwidth]{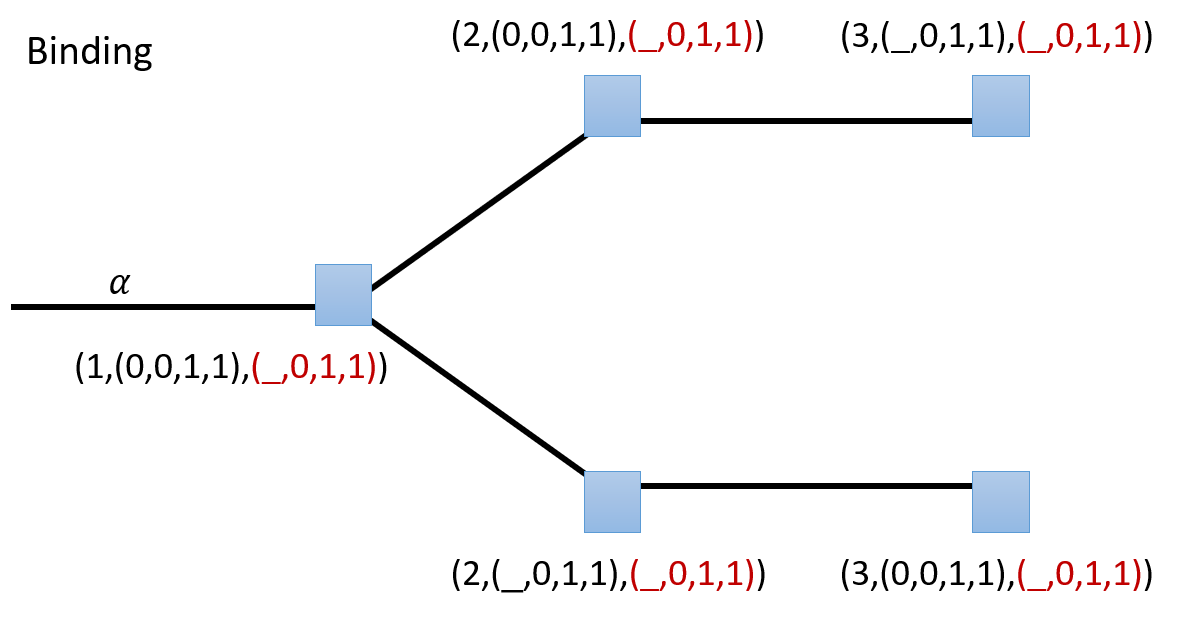}
    \includegraphics[width=0.45\textwidth]{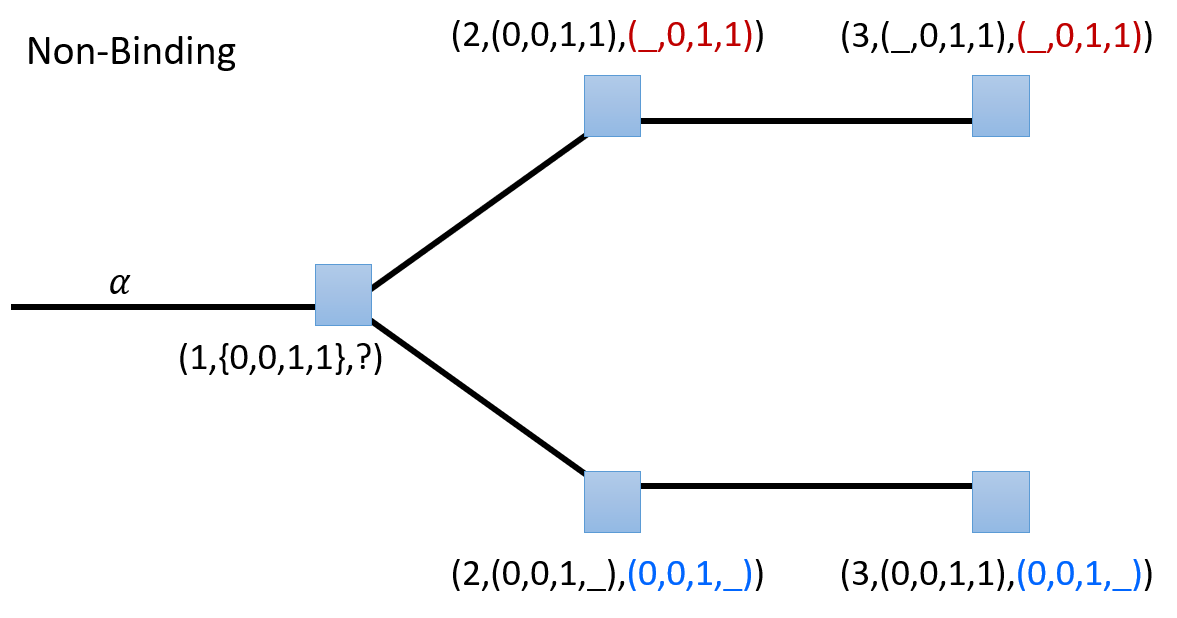}
\caption{
Binding vs.\ non-binding gather examples for $f = 1$ and $n = 4$.  
The processes are $p_1, \ldots, p_4$ with $p_4$ being faulty.  
Black lines indicate executions with time increasing to the right.
A square indicates a decision by a process and is labeled with the
process' id, the process' output, and the common core, 
which is ``?'' if it is not yet fixed.
The output and common core are denoted by a 4-vector with the $i$-th entry being the
value for $p_i$, or underscore if missing.
The common core consisting of the inputs of $p_2,p_3,p_4$ (resp., $p_1,p_2,p_3$) 
is colored red (resp. blue).}
\label{fig:binding gather}
\end{figure}

To demonstrate the binding property for gather, consider Figure~\ref{fig:binding gather},
in which $\alpha$ is an execution prefix that ends as soon as the
first correct process, $p_1$, returns from gather with output
$\langle 0,0,1,1 \rangle$.
The left-hand side of Figure~\ref{fig:binding gather} shows two different extensions of $\alpha$
when gather is binding.  
The common core is fixed to be $\langle \_,0,1,1 \rangle$ 
at the end of $\alpha$, and in both the upper
and lower extensions all the outputs contain the common core, even though
they are not all identical.
The right-hand side of Figure~\ref{fig:binding gather} shows two different extensions of $\alpha$
when gather is not binding.
The common core is not fixed at the end of $\alpha$.
In the upper extension, the common core is $\langle \_,0,1,1 \rangle$
but in the lower extension, the common core is $\langle 0,0,1,\_\rangle$.

Figure~\ref{fig:binding gather} also elucidates the interaction between the binding property
of gather and that of connected consensus in Algorithm~\ref{alg:cc 3f}
with $R = 1$ (crusader agreement).
Recall that Algorithm~\ref{alg:cc 3f} for $R =1$ is simply a call
to a gather subroutine followed by some local computation.
Suppose gather is binding and consider the top of Figure~\ref{fig:binding gather}.
The crusader agreement results in the upper extension are
$\bot$ for $p_1$ and $p_2$, and 1 for $p_3$,
while in the lower extension they are $\bot$ for $p_1$ and $p_3$,
and 1 for $p_2$; these decisions satisfy the binding property for
crusader agreement.
Now suppose gather is not binding and consider the bottom of Figure~\ref{fig:binding gather}.
The crusader agreement results in the upper extension are $\bot$ for
$p_1$ and $p_2$, and 1 for $p_3$,
while in the lower extension they are $\bot$ for $p_1$ and $p_3$,
and 0 for $p_2$.  
This violates the binding property for connected consensus
since in the upper extension, 1 is decided while in the lower extension 0 is decided.

The next theorem shows that Algorithm~\ref{alg:cc 3f} preserves the binding property of the gather subroutine used in it.
A binding gather algorithm appears in~\cite{SternA2024}
(see Appendix~\ref{app:gather}.)

\begin{theorem}
If the subroutine G satisfies binding for gather, then Algorithm~\ref{alg:cc 3f} satisfies binding for connected consensus.
\end{theorem}

\begin{proof}
We show an even stronger property for connected consensus:  as soon as the first correct process returns from its call to G in Line~\ref{line:cc 3f:first gather}, the branch for all decisions by correct processes in all extensions is fixed.

Let $\alpha$ be the prefix of any execution of the algorithm that ends
as soon as the first correct process returns from its first call to G.  
By the binding property of gather, the common core $S^C$ for the first
instance of G is fixed at the end of $\alpha$.
By Proposition~\ref{prop:cc 3f:gather core unique}, there is at most one
value in $V$ that occurs $|S^C| - f$ times in $S^C$.
Thus by Lemma~\ref{lem:cc 3f:gather result} there is at most one value in $V$ that occurs $|S_i| - f$ times in $S_i$,
for any correct process $p_i$, in any extension of $\alpha$.

Suppose $R = 1$.
Since the branch decided on is either $\bot$ or the value that
occurs $|S_i| - f$ times in $S_i$, the only possible non-$\bot$ branch
decision in any extension of $\alpha$ is the value (if any) occurring
$|S^C| - f$ times in $S^C$. 

Suppose $R \ge 2$.
Consider any extension $\alpha'$ of $\alpha$ in which a correct process
$p_i$ decides $(v,g)$ where $v \in V$.  
By the code (cf.\ Line~\ref{line:cc 3f:choose branch}), $v$ appears
$f+1$ times in $T_i$.
At least one of these occurrences corresponds to a correct process $p_j$.
By the validity property of gather, $v$ is $p_j$'s input to its second
call to G on Line~\ref{line:cc 3f:second gather} (the value of variable $v_j$).
By the code, $v$ appears $|S_i| - f$ times in $S_i$.
By Lemma~\ref{lem:cc 3f:gather result}, $v$ appears $|S^C| - f$ times
in $S^C$.
As argued above for $R = 1$, there is only one possible choice for $v$ and it is
fixed at the end of $\alpha$.
\end{proof}

\begin{remark}
The core set $S^C$ is hidden from the processes themselves.
When gather is not binding, the core set is determined only in hindsight, 
and could be captured as a prophecy variable. 
When gather is binding, the core set is determined once the first gather returns, and thus, it becomes a history variable.
(See~\cite{AbadiL1991} for a discussion of prophecy and history variables.)
\end{remark}

\subsection{Comparison with Other Connected Consensus Algorithms}

The result of instantiating Algorithm~\ref{alg:cc 3f} with a binding gather subroutine \cite{SternA2024} (see Appendix~\ref{app:gather})
and one of the optimally-resilient approximate agreement subroutines in~\cite{Coan1988,AbrahamAD2004}
is an optimally-resilient connected consensus algorithm for arbitrary $R$ that satisfies the binding property and has running time logarithmic in $R$.
(The running time is analyzed in
Appendix~\ref{app:gather}.)
To the best of our knowledge, it is the first such algorithm, 
since the algorithm in~\cite{ErbesW2024} does not satisfy binding while the algorithm in~\cite{Vazquez25} requires $n > 5f$ (and yet is not in canonical rounds).

Building blocks similar to connected consensus for arbitrary $R$  have been proposed in earlier work on \emph{synchronous} randomized algorithms for Byzantine agreement, for example $k$-proxcast in~\cite{ConsidineFFLMM2005} and $M$-gradecast in~\cite{GarayKKO2007}.
These papers typically make various cryptographic assumptions such as the existence of broadcast channels or digital signatures.

\section{Summary and Open Questions}
\label{chap:conclusion}

Our results show that canonical-round algorithms, while useful for structuring and reasoning about asynchronous algorithms, fundamentally fail to capture the communication patterns required for optimal-resilience Byzantine fault tolerance. 
Even when randomization is allowed\footnote{
Section~\ref{sec:lb} considers Las Vegas algorithms which have probabilistic termination.  Essentially the same argument works for Monte Carlo algorithms which have probabilistic safety properties, as long as the probability of validity, $\rho$, is more than $1/2$ and the probability of agreement is more than $2-2\rho$.  This subsumes the case when both properties are satisfied with error at most $1/3$.
Related results for general algorithms and $n \le 3f$ appear in~\cite{AbrahamN2021}.}, 
canonical-round algorithms require an unbounded number of rounds in the regime $3f < n \le 5f$, and communication-closed variants are unsolvable. This explains why known optimally-resilient Byzantine algorithms abandon strict round-based structure in favor of content-dependent primitives such as reliable broadcast and gather.

The gather primitive provides one concrete way to express the content-dependent communication patterns excluded by canonical rounds. 
We show that gather admits constant-time 
solutions with optimal resilience and supports simple reductions to tasks such as connected consensus, while preserving desirable properties such as binding.

Several open questions remain. 
First, the threshold $n = 5f$ marks a sharp boundary for bounded-round canonical-round algorithms in the Byzantine setting. 
It is natural to ask whether this threshold is strict:  
does every optimally-resilient algorithm admit a canonical-round formulation when $n > 5f$ as do randomized consensus~\cite{BenOr1983}, approximate agreement~\cite{DolevLPSW1986}, and connected consensus
(see Appendix~\ref{app:cc 5f})?
Second, it would be interesting to understand whether similar limitations arise in other adversarial models, such as partially synchronous systems or settings with authentication, where equivocation can be avoided by other means.

\section*{Acknowledgments}
Hagit Attiya and Itay Flam are supported by the Israel Science Foundation (22/1425 and 25/1849).

\bibliographystyle{plain}
\bibliography{references}

\appendix  

\section{(Binding) Gather Algorithm}
\label{app:gather}

We present an algorithm that solves the gather problem for $n > 3f$ in constant time,
achieving a common core of size $n-f$ and satisfying the binding property when required. 
The algorithm requires three communication phases for the non-binding case and four phases if binding is required. 
Since the algorithm is not in canonical rounds due to its use of reliable broadcast and content-dependent criteria to define communication phases, we measure running time instead of rounds.
When instantiated with Bracha's reliable broadcast primitive, the algorithm has worst-case running time of 7 time units for the non-binding variant and 9 time units for the binding variant.
Corollary~\ref{cor:gather-lb} justifies the fact that the algorithm is not in canonical-round form.

\subsection{The Algorithm}

The gather algorithm we describe is based on~\cite{SternA2024,SternA2021}.
Initially, each process disseminates its input value using an instance of a reliable broadcast primitive with itself as the designated sender; these are ``phase 1'' messages.
Processes then wait to accept $n-f$ phase 1 messages from the reliable broadcast instances.
As asynchrony can cause different processes to accept messages in different orders, no common core of size $n-f$ can be guaranteed yet and so processes proceed by exchanging messages over point-to-point channels, i.e., not via reliable broadcast.
Each process $p_i$ sends a ``phase 2'' message, which contains the set $T_i$ of (process id, value) pairs obtained from the first $n-f$ phase 1 messages it has accepted.
Process $p_i$ {\em approves} a phase 2 (or larger) message when it has also accepted (via reliable broadcast) all the values contained in the message; after approving $n-f$ phase 2 messages, it computes the union $U_i$ of all the sets in these messages.
At this point, 
a common core is still not guaranteed
so processes continue for another phase.
Process $p_i$ sends a ``phase 3'' message containing $U_i$ and
after approving $n-f$ phase 3 messages, it computes the union $V_i$ of all the sets in these message.
As shown in Lemma~\ref{lemma:gather:common_core}, a common core is now guaranteed.
However, the {\it binding} common core property is not guaranteed 
and requires one final phase.
Process $p_i$ sends a ``phase 4'' message containing $V_i$ and after approving $n-f$ phase 4 messages, it computes the union $W_i$ of all the sets in these messages.
Lemma~\ref{lemma:gather:binding} shows that the binding property is now ensured.

The pseudocode for the gather algorithm 
is presented in Algorithm~\ref{alg:binding_gather}.  
Three threads run concurrently on each process.  
One thread handles the acceptance of messages sent using the reliable broadcast instances.
Another thread handles the receipt of messages sent on the point-to-point channels.
The main thread starts when the algorithm is invoked.
Every time a message is accepted in the reliable broadcast thread or received in the point-to-point channels thread, the condition for the current wait-until statement in the main thread is evaluated.  
Thus, progress can be made in the main thread either when a reliable broadcast message is accepted, possibly causing more pairs to be accepted and thus more previously received messages to be approved, or when a point-to-point channel message is received, possibly causing the number of approved messages received to increase.

\begin{algorithm}[t]
\small
\begin{algorithmic}[1]{}
\Statex \Comment{Reliable Broadcast Acceptance Thread}
\end{algorithmic}

\begin{mdframed}[linecolor=black,linewidth=1pt,innerleftmargin=2pt,innerrightmargin=2pt,innertopmargin=3pt,innerbottommargin=3pt]
\begin{algorithmic}[1]
\State {\bf when} r-accept$(\langle 1, x \rangle,p_j)$ occurs:
    \Comment{$p_j$ is sender}
\State \quad add $\langle j, x \rangle$ to AP$_i$ \Comment{set of \emph{accepted pairs}}
\end{algorithmic}
\end{mdframed}

\begin{algorithmic}[1]{}
\Statex \Comment{Point-to-Point Channel Message Receipt Thread}
\end{algorithmic}

\begin{mdframed}[linecolor=black,linewidth=1pt,innerleftmargin=2pt,innerrightmargin=2pt,innertopmargin=3pt,innerbottommargin=3pt]
\begin{algorithmic}[1]
\setcounter{ALG@line}{2}
\State {\bf when} receive$(m)$ for sender $p_j$ occurs:
\State \quad add $m$ to RM$_i$ \Comment{set of \emph{received messages}}
\end{algorithmic}
\end{mdframed}

\begin{algorithmic}[1]{}
\Statex \Comment{Main Thread}
\end{algorithmic}

\begin{mdframed}[linecolor=black,linewidth=1pt,innerleftmargin=2pt,innerrightmargin=2pt,innertopmargin=3pt,innerbottommargin=3pt]
\begin{algorithmic}[1]
\setcounter{ALG@line}{4}
\Statex \emph{Terminology:}  a message $(r,X)$ is an \emph{approved phase $r$ message} if $X \subseteq AP_i$
\State {\bf when} \emph{gather}($x_i$, \emph{binding}) is invoked:
   \Comment{$x_i$ is $p_i$'s input, \emph{binding} is a Boolean}
\State \quad r-broadcast($\langle 1, x_i \rangle,p_i$) \Comment{initiate reliable broadcast instance with sender $p_i$}
\State \quad r-broadcast($-,p_j$) for all $j \ne i$   \Comment{and participate in instances with other senders}
\State \quad {\bf wait until} $|AP_i| = n-f$ \Comment{$n-f$ accepted pairs}
\State \quad $T_i \leftarrow AP_i$
\State \quad send $\langle 2, T_i \rangle$ to all processes \Comment{phase 2 message}
\State \quad {\bf wait until} $RM_i$ contains $n-f$ approved phase 2 messages
\State \quad $U_i \leftarrow \bigcup T_j$ such that $T_j$ is in an approved phase 2 message
            \label{line:gather compute U}
\State \quad send $\langle 3, U_i \rangle$ to all processes \Comment{phase 3 message}
\label{gather:line:wait_T_i}
\State \quad {\bf wait until} $RM_i$ contains $n-f$ approved phase 3 messages
\State \quad $V_i \leftarrow \bigcup U_j$ such that $U_j$ is in an approved phase 3 message
                       \label{line:gather:phase 3 end}
\State \quad {\bf if} $\neg${\it binding} {\bf then} {\bf return} $V_i$ 
                       \label{line:gather:non-binding return}
\State \quad{\bf else} send $\langle 4, V_i \rangle$ to all processes {\bf endif} \Comment{phase 4 message}
\State \quad {\bf wait until} $RM_i$ contains $n-f$ approved phase 4 messages
\State \quad $W_i \leftarrow \bigcup V_j$ such that $V_j$ is in an approved phase 4 message
\State \quad {\bf return} $W_i$
\end{algorithmic}
\end{mdframed}

\caption{Binding / non-binding gather using reliable broadcast ($n > 3f$); code for process $p_i$.}
\label{alg:binding_gather}
\end{algorithm}

\subsection{Correctness}

We show that the gather algorithm is correct for 
any $n > 3f$.

\begin{theorem}
\label{thm:gather:correctness}
Algorithm~\ref{alg:binding_gather} solves the gather problem and if the argument {\em binding} is true then it satisfies the binding property for gather.
\end{theorem}

\begin{proof}
Suppose the pair $\langle k,x \rangle$ is included in the output of a correct process $p_i$.
Then $p_i$ approved the pair after r-accepting it from the reliable broadcast.
If $p_k$ is correct, then by validity of reliable broadcast, $x$ is $p_k$'s input, and the validity property of gather holds.
Furthermore, if another correct process $p_j$ has the pair $\langle k,x' \rangle$ in its output, then it r-accepted that pair.  By agreement of reliable broadcast $x = x'$, and thus the agreement property of gather also holds.

We next argue progress through the phases of the algorithm. 
Since there are at least $n-f$ correct processes, each of which initially r-broadcasts its input, the validity property of reliable broadcast implies that every correct process eventually accepts at least $n-f$ pairs.
Thus each correct process sends a phase 2 message to all processes.

If any correct process $p_i$ sends $T_i$ in a phase 2 message, then it has accepted all pairs in $T_i$.
Thus, if another correct process $p_j$ receives $T_i$ in a phase 2 message from $p_i$, the totality property of reliable broadcast implies that $p_j$ eventually accepts all the pairs in $T_i$, and approves the phase 2 message from $p_i$ containing $T_i$.
This implies:

\begin{proposition}
\label{prop:gather:third message}
Every correct process eventually sends a phase 3 message.
\end{proposition}

By Proposition~\ref{prop:gather:third message} and an argument similar to the one proving it, we also have: 

\begin{proposition}
\label{prop:gather:fourth message}
If \emph{binding} is false, then every correct process eventually terminates; otherwise, it eventually sends a phase 4 message.
\end{proposition}

Finally, for the binding version of the algorithm, by Proposition~\ref{prop:gather:fourth message} and similar arguments:

\begin{proposition}
\label{prop:gather:termination}
If \emph{binding} is true, then every correct process eventually terminates.
\end{proposition}

The next lemma shows that the common core property holds for the sets $V_i$ of correct processes. 
Since the non-binding version of Algorithm~\ref{alg:binding_gather} terminates in Line~\ref{line:gather:non-binding return} and returns $V_i$, this implies that the common core property holds for that version.

\begin{lemma}\label{lemma:gather:common_core}
There exists a set $S^C$ of size $n-f$ that is contained in every set $V_i$ computed by a correct process $p_i$ in Line~\ref{line:gather:phase 3 end}.
\end{lemma}

\begin{proof}
We first argue that there is a correct process $p_j$ and a set of $f+1$ distinct correct processes $p_{i_0},\ldots,p_{i_f}$ (which might include $p_j$) such that $T_j \subseteq U_{i_k}$, for every $k, 0 \leq k \leq f$.

Let $G$ be the set consisting of the first $n-f$ correct processes that complete phase 3; we will show that $G$ must contain the desired $p_{i_0}$ through $p_{i_f}$. 
Let $m$ be the total number of phase 2 messages from correct processes that are approved by processes in $G$ before they send their phase 3 messages, counting duplicates (i.e., if both $p_i$ and $p_j$ approve a phase 2 message from $p_k$, count that as two messages).
Let $t \le f$ be the actual number of faulty processes in the execution.
Since each of the $n-f$ processes in $G$ approves $n-f$ phase 2 messages before sending its phase 3 message, at least $n-f-t$ of which are from correct processes, $m$ is at least $(n-f)(n-f-t)$.

Suppose in contradiction that there is no correct process whose
phase 2 message is approved by at least $f+1$ processes in $G$ before they send their phase 3 messages.
For each of the $n-t$ correct processes in the system, its phase 2 message is approved by at most $f$ processes in $G$ before they send their phase 3 messages.
Thus $m$ is at most $(n-t)f$.
\begin{align*}
(n-f)(n-f-t) &\le (n-t)f \\
\Rightarrow (n-f)(n-t) - (n-f)f &\le (n-t)f \\
\Rightarrow (n-f)(n-t) &\le (n-t)f + (n-f)f \\
\Rightarrow (n-f)(n-t) &\le (n-t)f + (n-t)f \quad\text{since $n-f \le n-t$} \\
\Rightarrow n-f &\le 2f \\
\Rightarrow n &\le 3f
\end{align*}
Contradiction.

Thus, the phase 2 message sent by at least one correct process, call it $p_j$, is approved by at least $f+1$ processes in $G$ during phase 3, call any $f+1$ of them $p_{i_0}$ through $p_{i_f}$. 
In other words, $T_j \subseteq U_{i_k}$, for every $k, 0 \le k \le f$.

In Line~\ref{line:gather:phase 3 end}, a correct process $p_i$ computes $V_i$ as the union of the sets of pairs appearing in the (at least) $n-f$ approved phase 3 messages it has received.
Since $(n-f) + (f+1) > n$, it is not possible for the senders of these $n-f$ approved 
phase 3 messages to be distinct from the $f+1$ processes $p_{i_0}$ through $p_{i_f}$. 
Thus at least one of the phase 3 messages approved by $p_i$ is from $p_{i_k}$ for some $k, 0 \le k \le f$, 
which implies that $U_{i_k} \subseteq V_i$.

Thus $T_j \subseteq U_{i_k} \subseteq V_i$, so setting $S^C$ equal to $T_j$ proves the lemma.
\end{proof}

We next proceed to show the binding property, when the \emph{binding} flag is true and the algorithm goes beyond Line~\ref{line:gather:non-binding return}. 
Note that the binding property encompasses the common core property.

\begin{lemma}
\label{lemma:gather:binding}
If \emph{binding} is true then Algorithm~\ref{alg:binding_gather} 
satisfies the binding property.
\end{lemma}

\begin{proof}
Let $\alpha$ be any execution prefix that ends when the first correct process $p_i$ decides, by outputting $W_i$.  
Before deciding, $p_i$ approves $n-f$ phase 4 messages, at least $n - 2f \ge f+1$ of which are from correct processes; choose exactly $f+1$ of these correct senders and denote them by $p_{i_0}, \ldots, p_{i_f}$. 

Let $S^C$ be the set of size $n-f$ contained in each of $V_{i_0}$ through $V_{i_f}$ (the contents of the phase 4 messages approved by $p_i$) whose existence is guaranteed by Lemma~\ref{lemma:gather:common_core}. 
We will show that $S^C$ is included in the decision of every correct process in every extension of $\alpha$.

Let $\alpha'$ be any extension of $\alpha$ and $p_j$ a correct process that decides in $\alpha'$, by outputting $W_j$.
By the code, $p_j$ approves $n-f$ phase 4 messages before deciding.
Since $(n-f) + (f+1) > n$, at least one of these approved phase 4 messages is from a correct process $p_{i_k}$, $0 \le k \le f$, one of the processes whose phase 4 message was approved by $p_i$ in $\alpha$.
Thus $S^C \subseteq V_{i_k} \subseteq W_j$.
\end{proof}

This completes the proof of the theorem.
\end{proof} 

We show that the common core property (even without binding) is not 
satisfied after phase 2. 
That is, the $U_i$ sets computed by the correct processes in Line~\ref{line:gather compute U} do not necessarily form a common core.

\begin{proposition}
\label{prop:gather:not common}
When $f = 2$ and $n = 7$, Algorithm~\ref{alg:binding_gather} does not ensure the common core property after phase 2.
\end{proposition}

\begin{proof}
Consider the following example.
Let $p_1,..,p_5$ be correct processes and $p_6, p_7$ be Byzantine.
Denote each process' input by its index (e.g. $p_1$'s input is $x_1 = 1$).
Table~\ref{table:no common core f=2} illustrates an order of events that results in a $U_i$ set for each correct process (for simplicity, we replace the pair $(i,i)$ with $i$ in the table).

\begin{table}[bt]
\begin{center}
\begin{tabular}{|c|c|c|c|c|c|c|c|}
\hline
 & $p_1$ & $p_2$ & $p_3$ & $p_4$ & $p_5$ & $p_6$ & $p_7$ \\
\hline
Input & 1 & 2 & 3 & 4 & 5 & 6 & 7 \\
\hline
$T_i$ set   & 1,4,5,6,7 & 2,4,5,6,7 & 3,4,5,6,7 & 2,3,4,6,7 & 1,4,5,6,7 & NA & NA \\
\hline
\makecell{Approved \\ $T_i$ sets}   & \makecell{ $T_1 \cup T_3 \cup $ \\ $T_5 \cup T_6 \cup$ \\ $T_7$ }  & \makecell{ $T_1 \cup T_2 \cup $ \\ $T_5 \cup T_6 \cup$ \\ $T_7$ } & \makecell{ $T_2 \cup T_3 \cup$ \\ $T_4 \cup T_6 \cup$ \\ $T_7$ } & \makecell{ $T_2 \cup T_3 \cup$ \\ $T_4 \cup T_6 \cup$ \\ $T_7$ } & \makecell{ $T_1 \cup T_2 \cup$ \\ $T_5 \cup T_6 \cup$ \\$ T_7$ } & NA & NA \\
\hline
\makecell{Resulting \\ $U_i$ sets} & 1,3,4,5,6,7 & 1,2,4,5,6,7 & 2,3,4,5,6,7 & 2,3,4,5,6,7 & 1,2,4,5,6,7 & NA & NA \\
\hline
\end{tabular}
\end{center}
    \caption{Counter-example with no common core before the third phase, for $n=7$, $f=2$.}
    \label{table:no common core f=2}
\end{table}

Since the adversary controls the scheduling, we can assume that each row in the table is executed in a ``linear'' manner. For example, each process r-broadcasts its input, then $p_1$ r-accepts messages from $p_4,p_5,p_6,p_7$ and itself (similarly for the other correct processes), and finally each correct process receives $n-f$ $T_j$ sets (in approved phase 2 messages) and immediately r-accepts any pairs included in these sets which it has not accepted so far, and thus approves\footnote{A set is approved if it is contained in a message that is approved.}
all received sets. 
Byzantine processes send their ``input'' to the correct processes via reliable broadcast, so if two correct process r-accept a pair from a Byzantine process, it is the same pair. The Byzantine processes can send any arbitrary $T_i$ set in a phase 2 message, so they send to each correct process $p_i$ a set that equals the correct process’ $T_i$ set, and therefore the correct processes immediately approve the sets sent by Byzantine processes. 

Were correct processes to decide after computing their $U_i$ sets in the example above, there wouldn't be a common core of size $n-f = 5$,
since the size of the intersection of $U_1$ through $U_5$ is only 4.
\end{proof}

\subsection{Special Case of One Faulty Process}
\label{sec:corner}

In this subsection we show that when $f = 1$,
the gather algorithm achieves a 
common core after phase 2 and a \emph{binding} common core after phase 3. 
This is one phase less than is needed in the general case when $f \ge 2$.

The next lemma implies that a common core is achieved after phase 2.

\begin{lemma}
\label{lemma:gather:corner common core}
When $f = 1$ and $n > 3$, Algorithm~\ref{alg:binding_gather} ensures that there exists a set $S^C$ of size $n-f=n-1$ that is contained in every set $U_i$ computed by a correct process $p_i$ in Line~\ref{line:gather compute U}.
\end{lemma}

\begin{proof}
We argue that the common core property is satisfied once every correct process $p_i$ approves $n-f = n-1$ phase 2 messages and computes $U_i$ in Line~\ref{line:gather compute U}. 
Since $U_i$ is comprised of phase 2 sets, each of size $n-f=n-1$, it follows that $|U_i|$ is either $n-1$ or $n$.
The common core size is $n-1$.

Assume in contradiction there is an execution with no common core.
Then there are two correct processes $p_i$ and $p_j$ such that $|U_i| = |U_j| = n-1$ but $U_i \ne U_j$.  
W.l.o.g., assume $U_i = \{1,\ldots,n-1\}$ and $U_j = \{2,\ldots,n\}$.
Every phase 2 message received by $p_i$ contains the set $\{1,\ldots,n-1\}$ and every phase 2 message received by $p_j$ contains the set $\{2,\ldots,n\}$.
At least $n-2$ of the senders of the phase 2 messages approved by $p_i$ (resp., $p_j$) are correct; let $A_i$ (resp., $A_j$) be any subset of these processes of size exactly $n-2$.
Since correct processes send phase 2 messages with the same content, $A_i \cap A_j = \emptyset$.  There must be at least one additional process to serve as the sender of the $(n-1)^{st}$ phase 2 messages approved by $p_i$ and $p_j$.
Thus $n \ge |A_i| + |A_j| + 1 = 2n - 3$, which implies $n \le 3$, a contradiction.
\end{proof}

However, the common core computed after phase 2 does not necessarily satisfy the binding property.  

\begin{proposition}
\label{prop:gather:corner not binding}
When $f = 1$,
Algorithm~\ref{alg:binding_gather} does not ensure the binding
common core property after phase 2.
\end{proposition}

\begin{proof}
Consider the following example for the case when $n = 4$. 
Suppose processes $p_1$, $p_2$, and $p_3$ are correct and process $p_4$ is Byzantine. 
Let $\alpha$ be the following execution prefix:
\begin{itemize}
\item Each process reliably broadcasts its phase 1 message.
\item $p_1$ accepts $1$, $2$, and $3$ and sends a phase 2 message for $\{1,2,3\}$.
\item $p_1$ accepts $4$.
\item $p_2$ accepts $2$, $3$, and  $4$ and sends a phase 2 message for $\{2,3,4\}$.
\item $p_1$ receives and approves phase 2 messages $\{1,2,3\}$ from $p_1$, $\{2,3,4\}$ from $p_2$, and $\{1,2,3\}$ from $p_4$.
\item $p_1$ returns $\{1,2,3,4\}$.
\end{itemize}
Now we consider two possible extensions of $\alpha$.

In $\alpha_1$:
\begin{itemize}
\item $p_3$ accepts $1$, $2$, and $3$ and sends a phase 2 message for $\{1,2,3\}$.
\item $p_2$ receives and approves phase 2 messages $\{1,2,3\}$ from $p_1$, $\{1,2,3\}$ from $p_3$, and $\{1,2,3\}$ from $p_4$.
\item $p_2$ returns $\{1,2,3\}$.
\end{itemize}
The common core in $\alpha.\alpha_1$ is $\{1,2,3\}$.

Here is a different extension of $\alpha$, call it $\alpha_2$:
\begin{itemize}
\item $p_3$ accepts $2$, $3$, $4$ and sends a phase 2 message for $\{2,3,4\}$.
\item $p_2$ receives and approves phase 2 messages $\{2,3,4\}$ from $p_2$, $\{2,3,4\}$ from $p_3$, and $\{2,3,4\}$ from $p_4$.
\item $p_2$ returns $\{2,3,4\}$.
\end{itemize}
The common core in $\alpha.\alpha_2$ is $\{2,3,4\}$, contradicting the binding property for gather.
\end{proof}

Finally we argue that after phase 3, the binding property for gather is guaranteed when $f = 1$ and $n > 3$.

\begin{lemma}
\label{lemma:gather:corner binding}
If $f = 1$, $n > 3$,
and the {\em binding} flag (input) is true then Algorithm~\ref{alg:binding_gather} satisfies the binding property for gather after 3 phases.
\end{lemma}

Lemma~\ref{lemma:gather:corner binding} is proved the same as Lemma~\ref{lemma:gather:binding} with these changes:
references to $V$ sets are replaced with references to $U$ sets,
references to $W$ sets are replaced with references to $V$ sets,
references to phase 4 are replaced with references to phase 3, and
references to Lemma~\ref{lemma:gather:common_core} are replaced with references to Lemma~\ref{lemma:gather:corner common core}.

\subsection{Time Complexity}
\label{sec:gather time}

We now analyze the worst-case running time of Algorithm~\ref{alg:binding_gather}. 
Recall that, for each execution, we measure the time that elapses between the point when the last correct process begins the algorithm and the point when the last correct process finishes the algorithm, after normalizing the delay of every message between correct processes as taking 1 time unit.  

To analyze the gather algorithm, we assume a black box reliable broadcast primitive which guarantees that the worst-case time for a correct process to accept the message from a correct sender is $T_{cor}$ ({\it cor} for correct sender) and the worst-case time that elapses between the message acceptance of two correct processes is $T_{rel}$ ({\it rel} for relay) even if the sender is Byzantine.

\begin{theorem} 
\label{thm:gather time}
If parameter {\it binding} is false, then Algorithm~\ref{alg:binding_gather} has worst-case running time $T_{cor} + 2 \cdot \max(1,T_{rel})$.
Otherwise it has worst-case running time $T_{cor} + 3 \cdot \max(1,T_{rel})$.
\end{theorem}

\begin{proof}
Every correct process starts the algorithm and invokes its instance of reliable broadcast by time 0.
Thus by time $T_{cor}$, every correct process has accepted pairs from all the $n-f$ correct processes and sends its phase 2 message.
By time $T_{cor} + 1$, every correct process has received phase 2 messages from all the $n-f$ correct processes.
It's possible that one of the pairs accepted by a correct process $p_i$ immediately before sending its phase 2 message is from a Byzantine process $p_k$; thus any other correct process $p_j$ also accepts the pair from $p_k$ by $T_{rel}$ time later.
It follows that every correct process approves $n-f$ phase 2 messages, and sends its phase 3 message, by time $T_{cor} + \max(1,T_{rel})$.

Similarly, we can argue that every correct process approves $n-f$ phase 3 messages and
either decides in the nonbinding case, or sends its phase 4 message in the binding case, by time $T_{cor} + 2 \cdot \max(1,T_{rel})$.

Finally, a similar argument shows that in the binding case, every correct process decides by time $T_{cor} + 3 \cdot \max(1,T_{rel})$.
\end{proof}

We next calculate the worst-case running time for Bracha's reliable
broadcast algorithm~\cite{Bracha1987} (see Algorithm~\ref{alg:bracha}). 
The proof is a timed analog of the liveness arguments in Theorem 12.18 of \cite{AttiyaW2004}.

\begin{lemma}
\label{lemma:bracha time}
For Bracha's reliable broadcast algorithm, $T_{cor} = 3$ and $T_{rel} = 2$.
\end{lemma}

\begin{proof}
Suppose the sender is correct and begins at time 0 by sending an {\it initial} message.
By time 1, every correct process receives the sender's {\it initial} message and sends its {\it echo} message if it has not already done so.
By time 2, every correct process receives {\it echo} messages from all the correct processes and, since $n-f \ge (n+f)/2$, sends its {\it ready} message if it has not already done so.
By time 3, every correct process receives {\it ready} messages from all the correct processes and, since $n-f \ge 2f+1$, it accepts the message.  
Thus $T_{cor} = 3$.

Now suppose that a correct process $p_i$ accepts the value $v$ from the sender (which may be Byzantine) at time $t$.
Thus $p_i$ has received at least $2f+1$ {\it ready} messages for $v$ by time $t$, and at least $f+1$ of them are from correct processes.
As a result, every correct process receives at least $f+1$ {\it ready} messages for $v$ by time $t+1$ and sends its {\it ready} message by time $t+1$.
As shown in Lemma 12.17 of \cite{AttiyaW2004}, this {\it ready} message is also for $v$.
Thus every correct process $p_j$ receives at least $n-f \ge 2f+1$ {\it ready} messages by time $t+2$ and accepts the value, implying that $T_{rel} = 2$.
\end{proof}

Combining Theorem~\ref{thm:gather time} and Lemma~\ref{lemma:bracha time}, we get:

\begin{corollary}
If Algorithm~\ref{alg:binding_gather} uses Bracha's reliable broadcast algorithm, then its worst-case running time in the nonbinding case is 7, while in the binding case it is 9.
\end{corollary}

If one prefers to measure running time from when the {\it first} correct process begins the algorithm, then these numbers would increase by 1.
The reason is that every correct process wakes up at most one time unit after the first one, due to the receipt of a message.

\section{Two-Input Approximate Agreement}
\label{app:2input-AA}

The use of a general approximate agreement subroutine in
Algorithm~\ref{alg:cc 3f} is a bit of overkill.
The algorithm only needs a solution to approximate agreement in
the special case where the inputs of the correct processes are
drawn from a known set of two values.
We show in Algorithm~\ref{alg:2-input aa}
that this problem can be solved with $n > 3f$
more simply than the optimally-resilient approximate agreement
algorithms in \cite{Coan1988,AbrahamAD2004}.

All of the inter-process communication in Algorithm~\ref{alg:2-input aa} is
encapsulated in a primitive we call \emph{pinch}.  It
ensures that if every correct process begins with an input drawn from a set
$\{a,b\}$, then every correct process returns with an output that
is a non-empty subset of $\{a,b\}$, such that
if two correct processes have 1-element outputs, then their outputs are the same (\emph{agreement})
and if a correct process has $c$ in its output, then some correct
process has input $c$ (\emph{validity}).
An algorithm for pinch is presented in Algorithm~\ref{alg:pinch}; note
that it is \emph{not} in canonical-round form (and in fact pinch is a nontrivial convergence problem).

In Algorithm~\ref{alg:2-input aa} each process repeatedly calls pinch,
initially with its original input, and then calculates a new input for
pinch to use in the next iteration by averaging the values obtained
from pinch.
Each iteration halves the maximum difference between processes'
current values.
We assume that each instantiation of the pinch subroutine
uniquely labels the messages sent so that there is no confusion;
also, once an instantiation returns, any messages received 
subsequently for that instantiation are ignored.

\begin{algorithm}[tb]
\caption{Two-input $\epsilon$-approximate agreement for $\{a,b\}$
where $a,b \in \mathbb{R}$, $a \le b$
($n > 3f$); code for process $p_i$.}

\label{alg:2-input aa}

\begin{algorithmic}[1]{}
\small
\State $y_i \leftarrow x_i$  \Comment{$x_i$ is $p_i$'s input}
\State \textbf{for} $k = 1$ to $\lceil \log_2 ((b-a)/\epsilon) \rceil$
       \label{line:2-input aa:start for}
\State \quad $A_i \leftarrow$ pinch$(y_i)$
       \label{line:2-input aa:pinch}
\State \quad \textbf{if} $A_i = \{c,c'\}$ for some $c$ and $c'$
             \textbf{then} $y_i \leftarrow (c+c')/2$
       \label{line:2-input aa:2 elements}
\State \quad \textbf{else} $y_i \leftarrow c$ where $A_i = \{c\}$ for some $c$
             \textbf{endif}
       \label{line:2-input aa:1 element}
\State \textbf{endfor}
       \label{line:2-input aa:end for}
\State \textbf{return} $y_i$
\end{algorithmic}
\end{algorithm}

  
\begin{algorithm}[tb]
\caption{Pinch ($n > 3f$); code for process $p_i$.}

\label{alg:pinch}

\begin{algorithmic}[1]{}
\small
\Statex \Comment{Main thread}
\end{algorithmic}

\begin{mdframed}[linecolor=black,linewidth=1pt,innerleftmargin=2pt,innerrightmargin=2pt,innertopmargin=3pt,innerbottommargin=3pt]
\begin{algorithmic}[1]
\small
    \State $A_i \leftarrow \emptyset$
    \State send $\langle${\sc echo1}$,x_i \rangle$ to all
              \Comment{$x_i$ is $p_i$'s input}
            \label{line:pinch:send echo1}
    \State \textbf{wait until} $|A_i| = 2$ or ($|A_i| = 1$ and received
               $n-f$ $\langle${\sc echo2}$,x \rangle$ messages for some $x$) 
           \label{line:pinch:wait}
    \State \textbf{return} $A_i$
\end{algorithmic}
\end{mdframed}

\begin{algorithmic}[1]{}
\small
\Statex \Comment{Thread for receiving {\sc echo1} messages}
\end{algorithmic}

\begin{mdframed}[linecolor=black,linewidth=1pt,innerleftmargin=2pt,innerrightmargin=2pt,innertopmargin=3pt,innerbottommargin=3pt]
\begin{algorithmic}[1]
\setcounter{ALG@line}{4}
\small
    \State \textbf{upon} receiving $\langle${\sc echo1}$, x \rangle$: 
    \State \quad \textbf{if} received $f+1$ $ \langle${\sc echo1}$,x \rangle$ 
               messages \textbf{then} 
    \State \quad \quad \textbf{if} haven't sent $\langle${\sc echo1}$,x \rangle$ 
               yet \textbf{then} send $\langle${\sc echo1}$,x \rangle$ to all 
               \textbf{endif}
          \label{line:pinch:relay echo1}
    \State \quad \textbf{elseif} received $n-f$ $ \langle${\sc echo1}$,x \rangle$ 
                        messages \textbf{then} 
    \State \quad \quad \textbf{if} haven't sent any $\langle${\sc echo2}$,* \rangle$
                  yet \textbf{then} send $\langle${\sc echo2}$,x \rangle$ to 
                   all \textbf{endif}
    \State \quad \quad add $x$ to $A_i$
               \label{line:pinch:add echo1s}
    \State \quad \textbf{endif}
\end{algorithmic}
\end{mdframed}

\begin{algorithmic}[1]{}
\small
\Statex \Comment{Thread for receiving {\sc echo2} messages}
\end{algorithmic}

\begin{mdframed}[linecolor=black,linewidth=1pt,innerleftmargin=2pt,innerrightmargin=2pt,innertopmargin=3pt,innerbottommargin=3pt]
\begin{algorithmic}[1]
\setcounter{ALG@line}{11}
\small
    \State \textbf{upon} receiving $\langle${\sc echo2}$, x \rangle$: 
    \State \quad \textbf{if} received $n-f$ $\langle${\sc echo2}$,x \rangle$ 
          \textbf{then} add $x$ to $A_i$ \textbf{endif}
     \label{line:pinch:add echo2s}
\end{algorithmic}
\end{mdframed}

\end{algorithm}

\begin{theorem}
Algorithm~\ref{alg:pinch} for pinch is correct, assuming every correct process'
input is chosen from a set of size 2.
\end{theorem}

\begin{proof}
Suppose correct process $p_i$ outputs $\{y_i\}$. 
Then by Line~\ref{line:pinch:wait}, $A_i = \{y_i\}$
and $p_i$ received $n-f$ {\sc echo2} messages for some $x$.
By Line~\ref{line:pinch:add echo2s}, $p_i$ added $x$ to $A_i$.
Thus $x = y_i$ and $n-f$ {\sc echo2} messages are sent for $y_i$.
Suppose another correct process $p_j$ outputs $\{y_j\}$.
By the same argument, $n-f$ {\sc echo2} messages are sent for $y_j$.
Since each correct process sends only one {\sc echo2} message
and $n > 3f$, $y_i$ must equal $y_j$.
This shows the agreement property for pinch.

We now show the validity property for pinch.
Suppose a value $x$ is in the output set of correct process $p_i$.
If $p_i$ adds $x$ to $A_i$ in Line~\ref{line:pinch:add echo1s}, then it 
receives $n-f$ {\sc echo1} messages for $x$.
At least $f+1$ of these are from correct processes.
Let $p_j$ be the first correct process to send an {\sc echo1} message
for $x$.
It cannot send the message in Line~\ref{line:pinch:relay echo1} since 
no correct process has yet sent that message.
Thus it sends the message in Line~\ref{line:pinch:send echo1} 
and $x$ is its input.

On the other hand, if $p_i$ adds $x$ to $A_i$ in 
Line~\ref{line:pinch:add echo2s},
then it receives $n-f$ {\sc echo2} messages for $x$.
At least $f+1$ of these messages are from correct processes,
including some $p_j$.
The reason $p_j$ sends an {\sc echo2} message for $x$ is that
it has received $n-f$ {\sc echo1} messages for $x$.
As argued in the previous paragraph, some correct process has
input $x$.
\end{proof}

The proof of correctness of Algorithm~\ref{alg:2-input aa} depends
on some invariants, stated in the next lemma.
Iteration 0 of the for loop refers to the point in the code immediately
before Line~\ref{line:2-input aa:start for}.

\begin{lemma}
\label{lem:2-input aa:invariants}
Assume each correct process begins Algorithm~\ref{alg:2-input aa} with
input either $a$ or $b$.
For each $k$,
$0 \le k \le \lceil \log_2(b-a)/\epsilon) \rceil$,
there exist real numbers $r$ and $r'$ in $[a,b]$ such that 
after iteration $k$ of the for loop in Algorithm~\ref{alg:2-input aa},\\
(I1) $|r - r'| \le (b-a)/2^k$; \\
(I2) for each correct process $p_i$, $y_i$ equals $r$ or $r'$; and \\
(I3) if all correct processes have the same input, say $a$,
then $y_i = a$ for every correct process $p_i$.
\end{lemma}

\begin{proof}
The proof is by induction on $k$.

\emph{Basis:}  $k = 0$.
Let $r = a$ and $r' = b$.

(I1) is true since $|a-b| \le (b-a)/2^0$.

(I2) is true by assumption that the input of each correct process is
either $a$ or $b$.

(I3) is true since $y_i$ is initialized to $p_i$'s input.

\emph{Induction:}  Assume for iteration $k-1$ and show for $k$.

By the inductive hypothesis, every correct process $p_i$ ends
iteration $k-1$ and starts iteration $k$ with $y_i$ equal to some $r$
or $r'$ such that $|r - r'| \le (b-a)/2^{k-1}$.
Thus each correct process $p_i$ calls pinch with argument $r$ or $r'$
and the output of
pinch, assigned to variable $A_i$, is either $\{r\}$, $\{r'\}$ or $\{r,r'\}$.
By the agreement property of pinch, there is at most one value,
say $r$, appearing in singleton outputs at correct processes.  By
Lines~\ref{line:2-input aa:2 elements}--\ref{line:2-input aa:1
element}, at the end of iteration $k$, $y_i$ is either $r$ or
$(r+r')/2$.  Thus (I2) holds for $k$.

To show (I1) for $k$:
\begin{align*}
|r - (r+r')/2| &= |r/2 - r'/2| = |r - r'| / 2 \\
               &\le ((b-a)/2^{k-1})/2 \\
               &= (b-a)/2^k. 
\end{align*}

To show (I3) for $k$, suppose all correct processes have the same input, $a$.
By the inductive hypothesis, they all end iteration $k-1$ with value $a$
and call pinch in iteration $k$ with argument $a$.
By the validity property for pinch, they all get $\{a\}$ as the output
of pinch, and thus $y_i = a$ for every correct process $p_i$.
\end{proof}

\begin{theorem}
Algorithm~\ref{alg:2-input aa} for 2-input $\epsilon$-approximate agreement is correct,
assuming every correct process' input is either $a$ or $b$.
\end{theorem}

\begin{proof}
\emph{Termination} holds by the assumed correctness, including termination,
of the pinch subroutine.

We now show \emph{$\epsilon$-agreement} holds.
After the last iteration $k = \lceil \log_2((b-a)/\epsilon) \rceil$,
the outputs of any two correct processes $p_i$ and $p_j$
are $y_i$ and $y_j$.
By Invariant (I2), $y_i$ and $y_j$ are both either $r$ or $r'$. 
By Invariant (I1),
\begin{align*}
|r - r'| &\le (b-a)/2^{\lceil \log_2((b-a)/\epsilon) \rceil} \\
     &\le (b-a)/2^{\log_2((b-a)/\epsilon)} \\
     &= (b-a)/((b-a)/\epsilon)) \\
     &= \epsilon.
\end{align*}

To show \emph{validity}, first note that if all the correct processes'
inputs are the same, say $a$, then by 
Invariant (I3), all correct processes' decisions are $a$.
Otherwise, validity holds since by Invariant (I2), every
correct process' decision is in $[a,b]$.
\end{proof}

\section{Canonical-Round Algorithm for Connected Consensus when $n > 5f$ }
\label{app:cc 5f}

We present a canonical-round algorithm for $R$-connected consensus
for any $R \ge 1$, assuming $n > 5f$.
(See Algorithm~\ref{alg:cc 5f}.)
It is an extension of Algorithm 2 in~\cite{AttiyaW2023} for $R = 1$.
To handle larger values of $R$, it uses an approximate agreement 
subroutine that is assumed to be canonical-round and that works for
$n > 5f$, such as that in~\cite{DolevLPSW1986}.
To our knowledge, it is the first connected consensus algorithm 
for arbitrary $R$ when $n > 5f$ that is in canonical-round form, 
as the one in~\cite{Vazquez25} uses reliable broadcast and a witness technique.

\begin{algorithm}[tb]
\caption{Byzantine-tolerant canonical-round $R$-connected consensus 
for arbitrary $R \ge 1$ ($n > 5f$); code for process $p_i$.  
AA is a canonical-round approximate agreement subroutine.}
\label{alg:cc 5f}

\begin{algorithmic}[1]{}
\small
\State send (1,$x_i$) to all 
           \Comment{round 1:  exchange branch candidates; $x_i$ is $p_i$'s input}
\State \textbf{wait} for $n-f$ round 1 messages;
        let $S_i$ be the multiset of values in received round 1 messages
\State \textbf{if} every element of \emph{trim}$(S_i,f)$ is some $v$ 
       \Comment{\emph{trim} drops $f$ largest and $f$ smallest elements of $S_i$}
\State \quad \textbf{then} $v_i \leftarrow v$; $g_i \leftarrow 1$ 
             \textbf{else} $v_i \leftarrow \bot$; $g_i \leftarrow 0$ 
             \textbf{endif}
         \label{line:cc 5f:round 1}
\State \textbf{if} $R = 1$
       \textbf{then} decide $(v_i,g_i)$
       \textbf{endif}
       \Comment{crusader agreement case; return after deciding}
\Statex \Comment{rest of code is for $R \ge 2$ case}
\State send $(2,v_i)$ to all
         \Comment{round 2: exchange branch candidates, now with only one 
                  non-$\bot$ value}
\State \textbf{wait} for $n-f$ round 2 messages; let $T_i$ be multiset
                     of values in received round 2 messages
\State \textbf{if} some $v \in V$ appears $f+1$ times in $T_i$
\Statex \quad\quad \textbf{then} $v_i \leftarrow v$
       \textbf{else} $v_i \leftarrow \bot$ \textbf{endif}
          \Comment{choose unique branch}
       \label{line:cc 5f:unique branch}
\State \textbf{if} some $v \in V$ appears $n-2f$ times in $T_i$
\Statex \quad\quad \textbf{then} $r_i \leftarrow R$ 
       \textbf{else} $r_i \leftarrow 0$ 
       \textbf{endif}
          \Comment{choose AA input}
       \label{line:cc 5f:choose AA input}
\State $g_i \leftarrow \lceil$AA$(1,r_i)\rceil$
      \Comment{$O(\log R)$ rounds to choose grade via AA with $\epsilon = 1$}
      \label{line:cc 5f:AA call}
\State \textbf{if} $g_i = 0$ 
       \textbf{then} decide $(\bot,0)$ 
       \textbf{else} decide $(v_i,g_i)$
       \textbf{endif}
       \Comment{will show $v_i \ne \bot$ when $g_i > 0$}
\end{algorithmic}
\end{algorithm}


We know the following about the results of round 1 from Lemma 7 
in~\cite{AttiyaW2023}, replacing variable name \emph{branch} with $v_i$:

\begin{lemma}[\cite{AttiyaW2023}]
\label{lem:cc 5f:AW2023}
Consider the values to which correct processes $p_i$ and $p_j$ set
their $v_i$ and $v_j$ variables in Line~\ref{line:cc 5f:round 1}.

\noindent
(a) If $v_i$ is set to $v \in V$,
then at least $n-4f$ correct processes have input $v$,
at least $n-3f$ correct processes have input at most $v$, and at
least $n-3f$ correct processes have input at least $v$.

\noindent
(b) If $v_i$ is set to $\bot$, then not all correct processes have the same 
input.

\noindent
(c) If $v_i$ is set to $u \in V$ and $v_j$ is set to $v \in V$, then $u = v$.
\end{lemma}

The next lemma implies that any non-$\bot$ branch chosen in 
Line~\ref{line:cc 5f:unique branch} is unique. 
It is analogous to Lemma~\ref{lem:cc 3f:unique branch}.

\begin{lemma}
\label{lem:cc 5f:unique branch}
If $v \in V$ occurs $f+1$ times in $T_i$ and $u \in V$ appears $f+1$ times
in $T_j$ for not necessarily distinct correct processes $p_i$ and $p_j$,
then $u = v$.
\end{lemma}

\begin{proof}
By part (c) of Lemma~\ref{lem:cc 5f:AW2023}, if any correct process
sets $v_k$ to a non-$\bot$ value in Line~\ref{line:cc 5f:round 1},
it sets it to the same non-$\bot$ value.
Thus every correct process that doesn't send $\bot$
sends the same non-$\bot$ value in its round 2 message.
At least one occurrence of $v$ in $T_i$ is from a correct process
and at least one occurrence of $u$ in $T_j$ is from a correct process.
Thus $v = u$.
\end{proof}

Next we show that the decisions are well-defined, in that the output
of any correct process is either $(\bot,0)$ or $(v,g)$ where $v \in V$
and $g > 0$.  
This lemma is analogous to Lemma~\ref{lem:cc 3f:well defined decision}.

\begin{lemma}
\label{lem:cc 5f:well defined decision}
If $R \ge 2$ and correct process $p_i$ decides $(v,g)$, then $g = 0$
if and only if $v = \bot$.
\end{lemma}

\begin{proof}
{\it Case 1:} $g_i$ is set to 0 in Line~\ref{line:cc 5f:AA call}.
Then by the code the decision is $(\bot,0)$ and so $v = \bot$.

{\it Case 2:} $g_i$ is set to a positive value in 
Line~\ref{line:cc 5f:AA call}.
By the validity property of approximate agreement, some correct
process $p_j$ calls AA with input $r_j = R$.
Thus there is some value $u \in V$ that appears $n-2f$ times in $T_i$.
At least $n-3f$ of these messages are from correct processes, and thus
at least $n-4f$ of these messages are also received by process $p_i$.
Since $n > 5f$, $n-4f \ge f+1$.
By Lemma~\ref{lem:cc 5f:unique branch}, $u$ is the only non-$\bot$
value that appears $f+1$ times in $T_i$.
Thus $p_i$ sets $v_i$ to $u$ in Line~\ref{line:cc 5f:unique branch}.
Since $u \ne \bot$ and $u$ is $v$, $v \ne \bot$.
\end{proof}

\begin{theorem}
Algorithm~\ref{alg:cc 5f} is correct for any value of $R \ge 1$
and satisfies binding.
\end{theorem}

\begin{proof}
The correctness and binding property are shown for $R = 1$
in~\cite{AttiyaW2023}.
Assume $R \ge 2$ for the remainder of the proof.

\emph{Termination} follows from the assumed correctness, including termination,
of the AA subroutine.

\emph{Validity:}  
Suppose correct process $p_i$ decides $(v,g)$.

(i) If $v \ne \bot$, we must show some correct process has input $v$.
By the code, $v$ is some value that appears $f+1$ times in $T_i$
(cf.\ Line~\ref{line:cc 5f:unique branch}).
At least one of these appearances of $v$ is from a correct process $p_j$.
Part (a) of Lemma~\ref{lem:cc 5f:AW2023} shows that some correct
process has input $v$.

(ii) If all correct processes have the same input, we must show $g = R$.
By (the contrapositive of) part (b) of Lemma~\ref{lem:cc 5f:AW2023},
since all correct processes have the same input, no correct process
has $v_i = \bot$ after round 1.  
By part (c) of Lemma~\ref{lem:cc 5f:AW2023},
they all send the same value $v$ in round 2. 
Thus every correct process receives at least $n-2f$ round 2 messages
for $v$ and sets its $r_i$ variable to $R$ in Line~\ref{line:cc 5f:choose AA input}.
By the validity property of approximate agreement, each correct process 
gets $R$ as the output of AA and decides $(v,R)$.
Thus the $g$ component in $p_i$'s decision equals $R$.

\emph{Agreement:}
Suppose correct process $p_i$ decides $(u,g)$ and correct process $p_j$
decides $(v,h)$.
By Lemma~\ref{lem:cc 5f:well defined decision}, 
$g = 0$ if and only if $u = \bot$, and $h = 0$ if and only if $v = \bot$.

If $g$ and $h$ are both 0, we have $(u,g) = (v,h)$.

Suppose exactly one of $g$ and $h$ is 0, say $g$.
By the 1-agreement property of approximate agreement, $h = 1$.
Thus $(u,g) = (\bot,0)$ and $(v,h) = (v,1)$ are neighbors.

Suppose neither $g$ nor $h$ is 0.
By the 1-agreement property of approximate agreement, $|g-h| \le 1$.
Since $u$ appears $f+1$ times in $T_i$ and $v$ appears $f+1$ times
in $T_j$ (cf.\ Line~\ref{line:cc 5f:unique branch}), 
Lemma~\ref{lem:cc 5f:unique branch} implies $u = v$.
Thus $(u,g)$ and $(v,h)$ label the same or neighboring vertices
on the branch for $v = u$.

\emph{Binding:} The algorithm satisfies an even stronger property of
``built-in'' binding (cf.\ \cite{AttiyaW2023}): the inputs alone
determine the branch along which the decisions occur.  Fix an
assignment of inputs and assume there is one execution from that
initial configuration $I$ in which a correct process $p_i$ sets its
$v_i$ variable to $u$ and another in which a correct process $p_j$
(which might or might not equal $p_i$) sets its $v_j$ variable to
$v$; WLOG, asume $u < v$.  By part (a) of Lemma~\ref{lem:cc 5f:AW2023}, in the
first execution at least $n-3f$ correct processes have inputs that are
at most $u$, while in the second at least $n-3f$ different correct
processes have inputs that are at least $v$.  Thus $n$ is at least
$2(n-3f)$ plus the $f$ faulty processes, which contradicts the
assumption that $n > 5f$.  Thus there exists $v \in V$ such that for
all executions from $I$, no value other than $v$ can appear in correct
processes' $v_i$ variables at the end of round 1.  Since no
correct process $p_i$ can get more than $f$ round 2 messages for any $u
\in V$ other than $v$, $p_i$ cannot assign $u$ to $v_i$
and cannot decide $(u,r)$ for any $r$.
\end{proof}

As in Appendix~\ref{app:2input-AA}, a general approximate agreement
algorithm is not needed inside Algorithm~\ref{alg:cc 5f}, since every
correct process' input is from a known 2-element set.  We can use the
same 2-input approximate agreement algorithm for the $n > 5f$ case as
in the $n > 3f$ case (Algorithm~\ref{alg:2-input aa}), but replace
the pinch subroutine (Algorithm~\ref{alg:pinch}) with one that is simpler and in canonical-round
form; see Algorithm~\ref{alg:pinch 5f}.

\begin{algorithm}[tb]
\caption{Pinch algorithm with $n > 5f$; code for process $p_i$.}

\label{alg:pinch 5f}

\begin{algorithmic}[1]{}
\small
\State $A_i \leftarrow \emptyset$
\State send (1,$x_i$) to all \Comment{$x_i$ is $p_i$'s input}
\State wait for $n-f$ round 1 messages
\State $A_i \leftarrow \{x | x$ appears at least $f+1$ times in received
         round 1 messages$\}$
\State \textbf{return} $A_i$
\end{algorithmic}
\end{algorithm}

\begin{theorem}
Algorithm~\ref{alg:pinch 5f} for pinch is correct, assuming every correct 
process' input is chosen from a set of size 2.
\end{theorem}

\begin{proof}
We first show the validity property for pinch.
Suppose a value $x$ is in the output set of correct process $p_i$.
Then $p_i$ received $n-f$ round 1 messages for $x$, at least one of
which is from a correct process $p_j$.
Thus $x$ is the input of $p_j$.

We argue the agreement property by contradiction.
Suppose correct process $p_i$ outputs $\{x_i\}$ and correct process
$p_j$ outputs $\{x_j\}$, where $x_i \ne x_j$.
By validity, $x_i$ and $x_j$ are inputs of correct processes,
and no other value can be the input of a correct process.
By the code, $p_i$ receives $n-f$ round 1 messages,
at least $f+1$ of them are for $x_i$,
at most $f$ of them are for $x_j$, and $m$ of them are for other values.
Since only faulty processes can send values other than $x_i$ and $x_j$,
$m \le f$.
The number of round 1 messages received by $p_i$ for $x$ is
$n-2f-m$.
Since $m$ faulty processes are budgeted as the senders of values other
than $x_i$ and $x_j$,
the remaining number of faulty processes that can send $x_i$ is
at most $f-m$.
Thus at least $(n-2f-m)-(f-m) = n-3f$ of the round 1 messages for $x_i$ 
received by $p_i$ are from correct processes.

Process $p_j$ receives round 1 messages from at least $n-4f$ of
the correct processes from which $p_i$ received round 1 messages for $x_i$.
Since $n > 5f$, it follows that $n-4f \ge f+1$, and thus $p_j$ adds
$x_i$ to $A_j$.
This contradicts the assumption that $p_j$'s output is $\{x_j\}$.
\end{proof}

\end{document}